\documentclass[11pt]{article}

\newcommand{\cmp}{Comm. Math. Phys.~}

\newcommand{\jfa}{J. Funct. Anal.~}
\newcommand{\jpa}{J. Phys. A~}

\newcommand{\pra}{Phys. Rev. A~}

\usepackage{young}
\usepackage[latin1]{inputenc}
\usepackage{appendix}
\usepackage{epstopdf}
\usepackage{xcolor}
\usepackage{subfigure}
\usepackage[T1]{fontenc}
\usepackage[sc]{mathpazo}
\usepackage{amsmath}
\usepackage{amssymb}
\usepackage{graphicx,color}
\usepackage{framed}
\usepackage{multirow}
\usepackage{enumerate}
\usepackage{amsthm}
\usepackage{amsfonts,mathrsfs}
\usepackage{geometry} 
\usepackage{eepic}
\usepackage{ifthen}
\usepackage[vcentermath]{youngtab}
\usepackage[unicode=true,pdfusetitle, bookmarks=true,bookmarksnumbered=false,bookmarksopen=false, breaklinks=false,pdfborder={0 0 0},backref=false,colorlinks=false] {hyperref}
\hypersetup{
colorlinks,linkcolor=myurlcolor,citecolor=myurlcolor,urlcolor=myurlcolor}
\definecolor{myurlcolor}{rgb}{0,0,0.7}

\usepackage{color}

\newcommand{\blue}{\textcolor{blue}}

\newcommand{\proj}[1]{| #1\rangle\!\langle #1 |}

\usepackage{hyperref}
\hypersetup{pdfpagemode=UseNone}

\newcommand{\tinyspace}{\mspace{1mu}}

\newcommand{\op}[1]{\operatorname{#1}}

\newcommand{\abs}[1]{\left\lvert\tinyspace #1 \tinyspace\right\rvert}

\newcommand{\norm}[1]{\left\lVert\tinyspace #1 \tinyspace\right\rVert}

\renewcommand{\det}{\operatorname{det}}
\renewcommand{\t}{{\scriptscriptstyle\mathsf{T}}}

\newcommand{\setft}[1]{\mathrm{#1}}

\newcommand{\density}[1]{\setft{D}\left(#1\right)}

\newcommand{\sep}[1]{\setft{Sep}\left(#1\right)}

\newcommand{\im}{\op{im}}

\newcommand{\sign}{\op{sign}}

\def\SO{\mathsf{SO}}
\def\SU{\mathsf{SU}}

\def\h{\mathfrak{h}}
\def\k{\mathfrak{k}}

\def\liet{\mathfrak{t}}
\def\su{\mathfrak{su}}

\def\Ad{\mathrm{Ad}}

\def\vol{\mathrm{vol}}
\def\dh{\mathrm{DH}}

\def\sep{\mathrm{sep}}
\def\haar{\mathrm{Haar}}
\def\Res{\mathrm{Res}}

\def \dif {\mathrm{d}}
\def \diag {\mathrm{diag}}
\def \vol {\mathrm{vol}}
\def \re {\mathrm{Re}}
\def \im {\mathrm{Im}}

\def\I{\mathbb{1}}

\def\zero{\mathbf{0}}
\def\Pol{\mathrm{Pol}}

\newenvironment{mylist}[1]{\begin{list}{}{
    \setlength{\leftmargin}{#1}
    \setlength{\rightmargin}{0mm}
    \setlength{\labelsep}{2mm}
    \setlength{\labelwidth}{8mm}
    \setlength{\itemsep}{0mm}}}
    {\end{list}}

\def\ot{\otimes}

\newcommand{\Inner}[2]{\left\langle #1 , #2\right\rangle}
\newcommand{\Innerm}[3]{\left\langle #1 \left| #2 \right| #3 \right\rangle}
\newcommand{\defeq}{\stackrel{\smash{\textnormal{\tiny def}}}{=}}

\newcommand{\Herm}{\mathrm{Herm}}

\newcommand{\pa}[1]{(#1)}
\newcommand{\Pa}[1]{\left(#1\right)}

\newcommand{\Br}[1]{\left[#1\right]}
\newcommand{\set}[1]{\{#1\}}
\newcommand{\Set}[1]{\left\{#1\right\}}

\DeclareMathOperator{\trace}{Tr}
\newcommand{\ptr}[2]{\trace_{#1}\pa{#2}}
\newcommand{\Ptr}[2]{\trace_{#1}\Pa{#2}}

\newcommand{\Tr}[1]{\Ptr{}{#1}}

\newcommand{\Abs}[1]{\left|\tinyspace#1\tinyspace\right|}

\def\cC{\mathcal{C}}\def\cE{\mathcal{E}}
\def\cF{\mathcal{F}}\def\cG{\mathcal{G}}\def\cH{\mathcal{H}}\def\cI{\mathcal{I}}
\def\cO{\mathcal{O}}

\def\cU{\mathcal{U}}

\def\bsA{\boldsymbol{A}}\def\bsB{\boldsymbol{B}}\def\bsC{\boldsymbol{C}}\def\bsD{\boldsymbol{D}}\def\bsE{\boldsymbol{E}}
\def\bsH{\boldsymbol{H}}
\def\bsM{\boldsymbol{M}}\def\bsO{\boldsymbol{O}}
\def\bsT{\boldsymbol{T}}
\def\bsU{\boldsymbol{U}}\def\bsV{\boldsymbol{V}}\def\bsX{\boldsymbol{X}}\def\bsY{\boldsymbol{Y}}
\def\bsZ{\boldsymbol{Z}}

\def\bsa{\boldsymbol{a}}\def\bsb{\boldsymbol{b}}\def\bse{\boldsymbol{e}}
\def\bsh{\boldsymbol{h}}

\def\bst{\boldsymbol{t}}
\def\bsx{\boldsymbol{x}}

\def\rD{\mathrm{D}}
\def\rH{\mathrm{H}}
\def\rK{\mathrm{K}}
\def\rS{\mathrm{S}}

\def\bbC{\mathbb{C}}

\def\bbR{\mathbb{R}}\def\bbT{\mathbb{T}}

\def\sfU{\mathsf{U}}

\newtheorem{thrm}{Theorem}[section]

\newtheorem{prop}[thrm]{Proposition}

\theoremstyle{definition}
\newtheorem{definition}[thrm]{Definition}
\newtheorem{remark}[thrm]{Remark}
\newtheorem{exam}[thrm]{Example}

\numberwithin{equation}{section}

\newcounter{questionnumber}

\begin{document}

\title{One application of Duistermaat-Heckman measure in quantum information theory}

\author{\blue{Lin Zhang}\footnote{E-mail: godyalin@163.com}, \quad\blue{Xiaohan Jiang},\quad \blue{Bing Xie}\\
  {\it\small School of Science, Hangzhou Dianzi University, Hangzhou 310018, PR~China}}
\date{}
\maketitle

\begin{abstract}
While the exact separability probability of 8/33 for two-qubit
states under the Hilbert-Schmidt measure has been reported by Huong
and Khoi
[\href{https://doi.org/10.1088/1751-8121/ad8493}{J.Phys.A:Math.Theor.{\bf57},
445304(2024)}], detailed derivations remain inaccessible for general
audiences. This paper provides a comprehensive, self-contained
derivation of this result, elucidating the underlying geometric and
probabilistic structures. We achieve this by developing a framework
centered on the computation of Hilbert-Schmidt volumes for key
components: the quantum state space, relevant flag manifolds, and
regular (co)adjoint orbits. Crucially, we establish and leverage the
connection between these Hilbert-Schmidt volumes and the symplectic
volumes of the corresponding regular co-adjoint orbits, formalized
through the Duistermaat-Heckman measure. By meticulously
synthesizing these volume computations --- specifically, the ratios
defining the relevant probability measures --- we reconstruct and
rigorously verify the 8/33 separability probability. Our approach
offers a transparent pathway to this fundamental constant, detailing
the interplay between symplectic geometry, representation theory,
and quantum probability.
\end{abstract}

\newpage\tableofcontents\newpage

\section{Introduction}


Quantum entanglement, a hallmark of quantum mechanics, is a
fundamental resource for quantum information processing.
Distinguishing entangled states from separable (non-entangled)
states is therefore a central problem in quantum information theory.
For bipartite systems, particularly the simplest non-trivial case of
two-qubits, determining the likelihood of encountering a separable
state within the ensemble of all possible states is a question of
significant theoretical interest. This is formalized as the
\emph{separability probability}
\cite{Zyczkowski1998a,Zyczkowski1998b}.

Under the \emph{Hilbert-Schmidt measure} --- a natural metric
induced by the Hilbert-Schmidt distance on the space of density
matrices \cite{Zyczkowski2002}
--- the exact value of the separability probability for two-qubit
states has been a subject of extensive numerical investigation and
conjecture. Early numerical studies strongly suggested the fraction
8/33 as the exact value (with 29/64 for the real two-qubit case)
\cite{Slater2007,Slater2012,Slater2013}. This conjecture was
rigorously confirmed recently by Huong and Khoi
\cite{Huong2024}(separately, the fraction 29/64 was analytically
proved by Lovas and Andai \cite{Lovas2017}). However, while their
result establishes the exact separability probability, the detailed
mathematical derivation presented in \cite{Huong2024} remains
inaccessible to a significant portion of the community.

This paper aims to bridge this gap in accessibility and
understanding. Our primary contribution is to provide a detailed,
self-contained, and pedagogically clear derivation of the
8/33\footnote{The derivation of the general value 8/33 is more
complex than that of the specific real two-qubit case 29/64.}
separability probability for two-qubit states under the
Hilbert-Schmidt measure. We achieve this by leveraging a powerful
geometric framework centered on the computation of volumes
\cite{Andai2006} in the relevant state spaces and their
substructures. The core of the used approach lies in the systematic
computation of Hilbert-Schmidt/symplectic volumes associated with
key geometric objects:
\begin{enumerate}[(1)]
\item \textbf{Flag manifold:} This is a specific instance constructed by taking the quotient of the unitary group $\mathsf{U}(N)$ by its maximal torus.
\item \textbf{A regular adjoint orbit}:  It is an orbit of a fixed diagonal matrix with distinct diagonal entries under the conjugate action of unitary
group $\sfU(N)$.  Note that this fixed diagonal matrix can be
translated to a regular element in the positive Weyl chamber for the
unitary group $\sfU(N)$.
\item \textbf{The space of all qudit states:} This is the compact
manifold of all density matrices acting on the same underlying space
$\bbC^N$.
\item \textbf{A regular co-adjoint orbit:} It is a orbit of a
regular element in the positive Weyl chamber under the co-adjoint
action of the unitary group $\sfU(N)$.
\end{enumerate}
A critical insight underpinning the derivation is the profound
connection between the Hilbert-Schmidt volumes of adjoint orbits and
the symplectic volumes of their corresponding regular co-adjoint
orbits. This connection is rigorously established through the
\emph{Duistermaat-Heckman (DH) measure} \cite{DH1982}, a fundamental
tool in symplectic geometry and geometric quantization that
intrinsically relates these volume measures. The DH measure's
properties are well-documented, with recent generalizations further
extending its scope \cite{Crooks2023,Crooks2024}. Crucially, its
density relative to the Lebesgue measure on the dual Lie algebra is
piecewise polynomial, computable via the Boylal-Vergne-Paradan jump
formula \cite{Boysal2009}. Moreover, the DH formalism has proven
instrumental in analyzing joint eigenvalue distributions of marginal
states in random multipartite quantum systems \cite{Christandl2014},
demonstrating its utility in quantum information theory.


The remainder of this paper is structured as follows: In
Section~\ref{sect:2}, we review necessary background on the
Hilbert-Schmidt metric, the geometry of the quantum state space, and
separability criteria. We provide details about the computation of
the Hilbert-Schmidt volume of the full quantum state space in
Section~\ref{sect:3} by calculating volume of flag manifolds and
(co)adjoint orbit. Section~\ref{sect:4} explores the connection to
symplectic geometry via co-adjoint orbits and the
Duistermaat-Heckman measure, deriving the corresponding symplectic
volumes and linking them back to the Hilbert-Schmidt framework. In
Section~\ref{sect:5}, we synthesizes the volume calculations from
the previous sections to compute the volume of the separable set and
derives the exact separability probability. Finally, in
Section~\ref{sect:6}, we conclude with a discussion of the
significance of the result and potential extensions.

\section{Preliminaries and notations}\label{sect:2}

Let the set of all $N\times N$ complex Hermitian matrices be denoted
by $\Herm(\bbC^N)$. Furthermore, $\Herm_{\mathrm{s}}(\bbC^N)$ denote
the set of all elements in $\Herm(\bbC^N)$ with simple spectrum
(i.e., those $N\times N$ complex Hermitian matrices with $N$
distinct eigenvalues).

\begin{prop}[\cite{Zhang2015,Deift2009}]\label{prop:1}
It holds that the subset $\Herm_{\mathrm{s}}(\bbC^N)$ is an open and
dense of full measure in $\Herm(\bbC^N)$. That is, the Lebesgue
measure of the complement is zero, i.e., the Lebesgue measure of
$\Herm(\bbC^N)\backslash \Herm_{\mathrm{s}}(\bbC^N)$ is zero.
\end{prop}

\begin{proof}
The technical proof is omitted here.
\end{proof}

We will denote by $\cH_N=\cH_1\ot\cH_2$ the $N$-dimensional Hilbert
space that describes the composite system of two components
$\cH_1,\cH_2$ with $\dim\cH_1=n, \dim\cH_2=m$ and $N=nm$. The set
$\rD(\bbC^n\ot\bbC^m)$ of bipartite states on $\cH_N$ can be
represented by complex $N\times N$ positive semi-definite matrices
of unit trace. Let
\begin{eqnarray}
\rD_{\mathrm{s}}(\bbC^n\ot\bbC^m):=\rD(\bbC^n\ot\bbC^m)\cap
\Herm_{\mathrm{s}}(\bbC^n\ot\bbC^m).
\end{eqnarray}
From Proposition~\ref{prop:1}, we see easily that
$\rD_{\mathrm{s}}(\bbC^n\ot\bbC^m)$ is also open and dense of full
measure in $\rD(\bbC^n\ot\bbC^m)$. A bipartite state $\rho_{AB}$ in
$\rD(\bbC^n \ot\bbC^m)$ is called \emph{separable} if and only if it
can be written as a convex combination of local product states:
$\rho_{AB} = \sum^K_{k=1} p_k \rho_{A,k} \ot \rho_{B,k}$ for some
$\rho_{A,k} \in \rD(\bbC^n)$, $\rho_{B,k} \in \rD(\bbC^m)$, and a
probability vector $(p_1,\ldots,p_K)$.

A natural parametrization of $\rho\in\rD(\bbC^n\ot\bbC^m)$ makes use
of the traceless Hermitian generators of $\SU(n)$ and $\SU(m)$:
\begin{eqnarray}\label{eq:parametrization}
\rho =
\frac1{nm}\Pa{\I_{nm}+\sum^{n^2-1}_{i=1}a_i\bsA_i\ot\I_m+\I_n\ot
\sum^{m^2-1}_{j=1}b_j\bsB_j +
\sum^{n^2-1}_{k=1}\sum^{m^2-1}_{l=1}c_{kl}\bsA_k\ot\bsB_l}
\end{eqnarray}
where
\begin{itemize}
\item all $a_i,b_j,c_{kl}$ are in $\bbR$,
\item $\bsA_k$ and $\bsB_l$ are the Hermitian generators of the
$\SU(n)$ and $\SU(m)$, respectively, chosen to satisfy
$\Inner{\bsA_i}{\bsA_j}=2\delta_{ij}=\Inner{\bsB_i}{\bsB_j}$,
\item two reduced states are: $\rho_{\bbC^n}=\frac1n\Pa{\I_n+\sum^{n^2-1}_{i=1}a_i\bsA_i}$
and $\rho_{\bbC^m}=\frac1m\Pa{\I_m+\sum^{m^2-1}_{j=1}b_j\bsB_j}$,
\item the vector
\begin{eqnarray*}
\bsx=\Pa{a_1,\ldots,a_{n^2-1},b_1,\ldots,b_{m^2-1},c_{11},\ldots,c_{n^2-1,m^2-1}}^\t\in\bbR^{(nm)^2-1}
\end{eqnarray*}
completely determines the state $\rho\in\rD(\bbC^n\ot\bbC^m)$ and
vice versa. The positivity of density matrices restricts the
possible vectors $\bsx$ to a proper subset $D^{(n\times m)}$ of
$\bbR^{(nm)^2-1}$.
\end{itemize}
According to the parametrization Eq.~\eqref{eq:parametrization} of
bipartite states, the Hilbert-Schmidt distance
$d_{\rH\rS}(\rho',\rho):=\sqrt{\Inner{\rho'-\rho}{\rho'-\rho}_{\rH\rS}}$
(see below Eq.~\eqref{eq:HS}) induces a flat metric on $D^{(n\times
m)}$ because
\begin{eqnarray}
d_{\rH\rS}(\rho',\rho) =
\frac1{nm}\Pa{2m\sum^{n^2-1}_{i=1}(a'_i-a_i)^2 +
2n\sum^{m^2-1}_{j=1}(b'_j-b_j)^2 +
4\sum^{n^2-1}_{k=1}\sum^{m^2-1}_{l=1}(c'_{kl}-c_{kl})^2}^\frac12.
\end{eqnarray}
In particular, for $n=m=2$, we get that
\begin{eqnarray}
d_{\rH\rS}(\rho',\rho) = \frac12\Pa{\sum^3_{i=1}(a'_i-a_i)^2 +
\sum^3_{j=1}(b'_j-b_j)^2 +
\sum^3_{k=1}\sum^3_{l=1}(c'_{kl}-c_{kl})^2}^\frac12 =
\frac12d_{\mathrm{Euclid}}(\bsx',\bsx),
\end{eqnarray}
where $\bsx'=(\bsa',\bsb',c'_{kl})$ and $\bsx=(\bsa,\bsb,c_{kl})$.
In such case, up to an overall constant $\frac12$, the following
mapping is bijective and isometric:
\begin{eqnarray}
(\density{\bbC^2\ot\bbC^2},d_{\rH\rS})\cong
(D^{(2\times2)},d_{\mathrm{Euclid}}).
\end{eqnarray}

\section{Riemannian volumes of manifolds}\label{sect:3}

Recall that an $N$-dimensional oriented manifold $\bsM$ with a
pseudo-Riemannian metric $g$ has a standard volume form $\omega$,
known as the Riemannian volume form, whose expression in an oriented
chart $(x^1,\ldots,x^N)$ is given by
\begin{eqnarray}
\omega = \sqrt{\det(g)}\dif x^1\wedge\cdots\wedge\dif x^N,
\end{eqnarray}
corresponding to the line element (aka, arc length differential or
differential of arc length) $\dif s^2 = \sum^N_{i,j=1}\dif x^i
g_{ij}\dif x^j$. If $D$ is a domain of integration on $\bsM$, then
\begin{eqnarray}
\vol_g(D) :=\int_{D}\omega
\end{eqnarray}
is called the \emph{Riemannian volume} of $D$.

In particular, on $\rD_{\mathrm{s}}(\bbC^N)$, we have the
Hilbert-Schmidt inner product, which is defined by
\begin{eqnarray}\label{eq:HS}
\Inner{\bsX}{\bsY} :=\Tr{\bsX^\dagger\bsY}.
\end{eqnarray}
Differentiating this inner product yields a metric on
$\rD_{\mathrm{s}}(\bbC^N)$ which we denote by $g_{\rH\rS}$. We shall
denote by $\vol_{\rH\rS}$ the Riemannian volume form associated with
$g_{\rH\rS}$ and refer to the volume measured by $\vol_{\rH\rS}$ as
the HS volume \cite{Zhang2018}.

\subsection{Hilbert-Schmidt volumes of flag manifolds}

A complex matrix $\bsZ$ can always be written as $\bsZ =
\bsX+\mathrm{i}\bsY$, where $\bsX=\re(\bsZ)$ and $\bsY=\im(\bsZ)$
are real matrices. Denote
\begin{eqnarray}
[\dif\bsZ]:=[\dif\bsX][\dif\bsY],
\end{eqnarray}
where $[\dif \bsX]$ means the product of independent differentials
in $\bsX$, and $[\dif\bsY]$ has a similar meaning.

Assume that $\bsU=(u_{ij})\in\sfU(N)$, where $u_{ij}\in\bbC$. Let
$u_{kk}=\abs{u_{kk}}e^{\mathrm{i}\theta_k}$ be its polar form, where
$\theta_k\in[-\pi,\pi)$. Then
\begin{eqnarray}
\bsU =
\Pa{\bsU\bsT^{-1}}\bsT,\quad\bsT=\diag\Pa{e^{\mathrm{i}\theta_1},\ldots,e^{\mathrm{i}\theta_N}}.
\end{eqnarray}
This means that any $\bsU\in\sfU(N)$ can be factorized into a
product of a unitary matrix with diagonal entries being nonnegative
and a diagonal unitary matrix. Thus, we have the following
diffeomorphism:
\begin{eqnarray}
\sfU(N)\cong \sfU(N)/\bbT^N\times \bbT^N,
\end{eqnarray}
where
$\bbT^N:=\Set{\diag\Pa{e^{\mathrm{i}\theta_1},\ldots,e^{\mathrm{i}\theta_N}}:\theta_k\in[-\pi,\pi),k=1,\ldots,N}$,
which is the maximal torus of $\sfU(N)$. For any measurable $f$ over
$\sfU(N)$,
\begin{eqnarray}
\int_{\sfU(N)}f(\bsU)[\bsU^\dagger\dif \bsU] =
\int_{\sfU(N)/\bbT^N}\int_{\bbT^N}f(\bsV\bsT)[\bsV^\dagger\dif
\bsV][\bsT^\dagger\dif\bsT],
\end{eqnarray}
where $\bsU=\bsV\bsT$ for $\bsV\in\sfU(N)/\bbT^N$ and
$\bsT\in\bbT^N$, and $[\bsU^\dagger\dif \bsU]=[\bsV^\dagger\dif
\bsV][\bsT^\dagger\dif \bsT]$, where $[\bsU^\dagger\dif \bsU]$ is a
left-invariant matrix-valued differential form, called Maurer-Cartan
form. Since $v_{kk}\in\bbR_{>0}$, and moreover $v_{kk}$ is not an
independent variable, it follows that
\begin{eqnarray}
[\bsV^\dagger\dif \bsV] &=& \prod_{1\leqslant i<j\leqslant
N}\re(\bsV^\dagger\dif\bsV)_{ij})\im(\bsV^\dagger\dif\bsV)_{ij}),\\{}
[\bsU^\dagger\dif \bsU] &=&
\prod^N_{k=1}\im((\bsU^\dagger\dif\bsU)_{kk})\times\prod_{1\leqslant
i<j\leqslant
N}\re(\bsU^\dagger\dif\bsU)_{ij}\im(\bsU^\dagger\dif\bsU)_{ij}
\end{eqnarray}
The Euclid volume of the flag manifold $\sfU(N)/\bbT^N$ is given by
\cite{Zhang2015}:
\begin{eqnarray}
\vol_{\mathrm{Euclid}}(\sfU(N)/\bbT^N) = \int_{\sfU(N)/\bbT^N}
[\bsV^\dagger\dif\bsV] =
\frac{\pi^{\binom{N}{2}}}{\prod^N_{k=1}\Gamma(k)}.
\end{eqnarray}
\begin{prop}[\cite{Zhang2015}]
The Hilbert-Schmidt volume of the flag manifold
$\mathsf{U}(N)/\bbT^N$ is given by
\begin{eqnarray}
\vol_{\rH\rS}(\mathsf{U}(N)/\bbT^N) = 2^{\binom{N}{2}}
\vol_{\mathrm{Euclid}}(\mathsf{U}(N)/\bbT^N) =
\frac{(2\pi)^{\binom{N}{2}}}{\prod^N_{k=1}\Gamma(k)},
\end{eqnarray}
\end{prop}

\begin{proof}
In fact, the line element for the Hilbert-Schmidt inner metric is
given by
$$
\dif s^2_{\rH\rS} = \Inner{\dif\bsV}{\dif\bsV}=
\Inner{\bsV^\dagger\dif\bsV}{\bsV^\dagger\dif\bsV},\quad
\bsV\in\sfU(N)/\bbT^N.
$$
It is clearly seen that $\bsV^\dagger\dif\bsV$ is skew-Hermitian and
its diagonal part is real. Thus the diagonal part of
$\bsV^\dagger\dif\bsV$ must be vanished, and
\begin{eqnarray*}
\dif s^2_{\rH\rS} &=& \sum^N_{k=1}[(\bsV^\dagger\dif\bsV)_{kk}]^2 +
2\sum_{1\leqslant i<j\leqslant
N}\Pa{[\re(\bsV^\dagger\dif\bsV)_{ij}]^2 +
[\im(\bsV^\dagger\dif\bsV)_{ij}]^2} \\
&=& 2\sum_{1\leqslant i<j\leqslant
N}\Pa{[\re(\bsV^\dagger\dif\bsV)_{ij}]^2 +
[\im(\bsV^\dagger\dif\bsV)_{ij}]^2}.
\end{eqnarray*}
The corresponding HS volume element is given by
\begin{eqnarray*}
\dif V_{\rH\rS} &=& \prod_{1\leqslant i<j\leqslant
N}\Pa{\sqrt{2}\re(\bsV^\dagger\dif\bsV)_{ij}\sqrt{2}\im(\bsV^\dagger\dif\bsV)_{ij}}\\
&=& 2^{\binom{N}{2}}\prod_{1\leqslant i<j\leqslant
N}\Pa{\re(\bsV^\dagger\dif\bsV)_{ij}\im(\bsV^\dagger\dif\bsV)_{ij}}
= 2^{\binom{N}{2}}[\bsV^\dagger\dif\bsV],
\end{eqnarray*}
which leads to the conclusion
\begin{eqnarray*}
\vol_{\rH\rS}\Pa{\sfU(N)/\bbT^N} &=& \int_{\sfU(N)/\bbT^N}\dif
V_{\rH\rS} = 2^{\binom{N}{2}}\int_{\sfU(N)/\bbT^N}[\bsV^\dagger\dif\bsV] \\
&=& 2^{\binom{N}{2}} \vol_{\mathrm{Euclid}}\Pa{\sfU(N)/\bbT^N}\\
&=&
2^{\binom{N}{2}}\frac{\pi^{\binom{N}{2}}}{\prod^N_{k=1}\Gamma(k)} =
\frac{(2\pi)^{\binom{N}{2}}}{\prod^N_{k=1}\Gamma(k)}.
\end{eqnarray*}
This completes the proof.
\end{proof}

\subsection{Hilbert-Schmidt volumes of regular adjoint orbits}

Denote the adjoint orbit $\cU_{\bsx}$ for
$\bsx=(x_1,\ldots,x_N)\in\bbR^N$, where $x_1>\cdots>x_N$, as
\begin{eqnarray}
\cU_{\bsx}:=
\Set{\bsU\diag(x_1,\ldots,x_N)\bsU^\dagger:\bsU\in\sfU(N)}.
\end{eqnarray}
Let $V_N(\bsx):=\prod_{1\leqslant i<j\leqslant N}(x_i-x_j)$. The
stablizer group of $\bsx$ is just the flag manifold
$\sfU(N)/\bbT^N$. Apparently, $\cU_{\bsx}\cong
\set{\bsx}\times\sfU(N)/\bbT^N$, where $\bbT^N=\sfU(1)^{\times N}$
is $N$-torus. Suppose that the unitary group $\sfU(N)$ is equipped
with Haar probability measure $\mu_{\haar}$. Consider the following
map
\begin{eqnarray*}
&\pi:& (\sfU(N),\mu_\haar)\longmapsto (\cU_{\bsx},\nu:=\pi_*\mu_\haar),\\
&&\bsU\mapsto\bsU\diag(x_1,\ldots,x_N)\bsU^\dagger.
\end{eqnarray*}
The push-forward measure of $\mu_\haar$ along the map $\pi$ is
$\nu:=\pi_*\mu_\haar$, called the \emph{orbital measure}
\cite{Olshanski2013}. So the integral along the orbit is given by
the following formula
\begin{eqnarray*}
&&\int_{\cU_{\bsx}}f\dif\nu = \Inner{\nu}{f} =
\Inner{\pi_*\mu_\haar}{f} = \Inner{\mu_{\haar}}{\pi^*f} \\
&&=
\int_{\sfU(N)}f\Pa{\bsU\diag(x_1,\ldots,x_N)\bsU^\dagger}\dif\mu_{\haar}(\bsU).
\end{eqnarray*}

%

The following result \cite{Collins2023} is mentioned without proof.
Here we provide detailed proof for the reference.
\begin{prop}\label{prop:HSvolumeoforbit}
The Hilbert-Schmidt volume of the adjoint orbit
$\cU_{\bsx}\cong\sfU(N)/\bbT^N$ is given by
\begin{eqnarray}
\vol_{\rH\rS}(\cU_{\bsx}) =
V_N(\bsx)^2\vol_{\rH\rS}(\sfU(N)/\bbT^N).
\end{eqnarray}
\end{prop}

\begin{proof}
Indeed, for $\bsX\in\cU_{\bsx}$, i.e.,
$\bsX=\bsV\diag(x_1,\ldots,x_N)\bsV^\dagger$, where
$\bsV\in\sfU(N)/\bbT^N$. Then
\begin{eqnarray*}
\dif\bsX=\dif\bsV \diag(x_1,\ldots,x_N)\bsV^\dagger +
\bsV\diag(x_1,\ldots,x_N)\dif\bsV^\dagger,
\end{eqnarray*}
which implies that
\begin{eqnarray*}
\bsV^\dagger(\dif\bsX)\bsV=\bsV^\dagger\dif\bsV
\diag(x_1,\ldots,x_N) + \diag(x_1,\ldots,x_N)\dif\bsV^\dagger\bsV.
\end{eqnarray*}
Due to the fact that $\bsV^\dagger\bsV=\I_N$, we infer that
$$
(\dif\bsV^\dagger)\bsV+\bsV^\dagger\dif\bsV=0\Longleftrightarrow
(\dif\bsV^\dagger)\bsV=-\bsV^\dagger\dif\bsV.
$$
This indicates that
$\bsV^\dagger(\dif\bsX)\bsV=\Br{\bsV^\dagger\dif\bsV,
\diag(x_1,\ldots,x_N)}$. Thus
\begin{eqnarray*}
[\dif\bsX] &=&
[\bsV^\dagger\dif\bsX\bsV]=\prod_{i<j}\Br{\bsV^\dagger\dif\bsV,
\diag(x_1,\ldots,x_N)}_{ij}\\
&=& \prod_{i<j}
(x_j-x_i)\re(\bsV^\dagger\dif\bsV)_{ij}(x_j-x_i)\im(\bsV^\dagger\dif\bsV)_{ij}\\
&=&V_N(\bsx)^2[\bsV^\dagger\dif\bsV],
\end{eqnarray*}
where $[\bsV^\dagger\dif\bsV]=\prod_{1\leqslant i<j\leqslant
N}\re(\bsV^\dagger\dif\bsV)_{ij}\im(\bsV^\dagger\dif\bsV)_{ij}$ is
the Euclid volume element on the flag manifold $\sfU(N)/\bbT^N$. We
can derive that $[\dif\bsX]=V_N(\bsx)^2[\bsV^\dagger\dif\bsV]$. Now
\begin{eqnarray*}
\vol_{\mathrm{Euclid}}(\cU_x)&=&\int_{\cU_{\bsx}}[\dif\bsX]=V_N(\bsx)^2\int_{\sfU(N)/\bbT^N}[\bsV^\dagger\dif\bsV]\\
&=&V_N(\bsx)^2\vol_{\mathrm{Euclid}}(\sfU(N)/\bbT^N).
\end{eqnarray*}
The differential of arc length is given by
\begin{eqnarray*}
\dif s^2_{\rH\rS} &=& \Inner{\dif\bsX}{\dif \bsX} = \Inner{[\bsV^\dagger\dif\bsV,\diag(x_1,\ldots,x_N)]}{[\bsV^\dagger\dif\bsV,\diag(x_1,\ldots,x_N)]}\\
&=&
\sum^N_{i,j=1}\abs{[\bsV^\dagger\dif\bsV,\diag(x_1,\ldots,x_N)]_{ij}}^2
=
\sum^N_{i,j=1}\abs{(x_j-x_i)(\bsV^\dagger\dif\bsV)_{ij}}^2\\
&=&\sum_{i\neq j}(x_j-x_i)^2\abs{(\bsV^\dagger\dif\bsV)_{ij}}^2 =
2\sum_{1\leqslant i<j\leqslant N}(x_j-x_i)^2\abs{(\bsV^\dagger\dif\bsV)_{ij}}^2\\
&=&\sum_{1\leqslant i<j\leqslant
N}\Pa{\sqrt{2}(x_j-x_i)\re(\bsV^\dagger\dif\bsV)_{ij}}^2+\sum_{1\leqslant
i<j\leqslant
N}\Pa{\sqrt{2}(x_j-x_i)\im(\bsV^\dagger\dif\bsV)_{ij}}^2.
\end{eqnarray*}
The corresponding HS volume is given by $\dif V_{\rH\rS} =
2^{\binom{N}{2}}V_N(\bsx)^2[\bsV^\dagger\dif\bsV]$. Thus
\begin{eqnarray*}
&&\vol_{\rH\rS}(\cU_{\bsx})=2^{\binom{N}{2}}\vol_{\mathrm{Euclid}}(\cU_{\bsx})=2^{\binom{N}{2}}V_N(\bsx)^2\vol_{\mathrm{Euclid}}(\sfU(N)/\bbT^N)\\
&&=2^{\binom{N}{2}}V_N(\bsx)^2
\frac{\pi^{\binom{N}{2}}}{\prod^N_{k=1}\Gamma(k)} =
V_N(\bsx)^2\frac{(2\pi)^{\binom{N}{2}}}{\prod^N_{k=1}\Gamma(k)},
\end{eqnarray*}
implying that the HS volume of adjoint orbit
$\cU_{\bsx}\cong\sfU(N)/\bbT^N$ is given by the following formula:
$$
\vol_{\rH\rS}(\cU_{\bsx}) =
V_N(\bsx)^2\frac{(2\pi)^{\binom{N}{2}}}{\prod^N_{k=1}\Gamma(k)}=V_N(\bsx)^2\vol_{\rH\rS}(\sfU(N)/\bbT^N).
$$
We are done.
\end{proof}

\subsection{Hilbert-Schmidt volumes of quantum state spaces }

Let
\begin{eqnarray}
\cC_N&:=&\Set{(x_1,\ldots,x_N)\in\bbR^N: x_1>\cdots>x_N},\\
\Delta_{N-1}&:=&
\Set{(x_1,\ldots,x_N)\in\bbR^N_{\geqslant0}:\sum^N_{k=1}x_k=1}.
\end{eqnarray}
The natural projection $\sfU(N)\to\sfU(N)/\bbT^N$, where
$\bbT^N=\set{\diag(e^{\mathrm{i}\theta_1},\ldots,e^{\mathrm{i}\theta_N}):\theta_1,\ldots,\theta_N\in\bbR}$
such that:
$$
\bsU\mapsto[\bsU]\equiv \bsU\bbT^N\in \sfU(N)/\bbT^N.
$$
Define the map $\tau:\Herm_{\mathrm{s}}(\bbC^N)\to \cC_N\times
\sfU(N)/\bbT^N$.

\begin{thrm}[\cite{Deift2009}]
It holds that the above map $\tau$ is a diffeomorphism:
\begin{eqnarray}
\Herm_{\mathrm{s}}(\bbC^N)\stackrel{\tau}{\cong}\cC_N\times
(\sfU(N)/\bbT^N).
\end{eqnarray}
Moreover, the map $\tau$ induces the following two diffeomorphisms:
\begin{eqnarray}
\rD_{\mathrm{s}}(\bbC^N)&\cong& (\Delta_{N-1}\cap\cC_N)\times
(\sfU(N)/\bbT^N),\\
\cU_{\bsx}&\cong& \set{\bsx}\times (\sfU(N)/\bbT^N)\cong
\sfU(N)/\bbT^N.
\end{eqnarray}
\end{thrm}

\begin{proof}
The proof is omitted here.
\end{proof}
Although the HS volume of the set of all density matrices acting on
$\bbC^N$ is derived already by \.{Z}yczkowski, we still include the
proof here for completeness.
\begin{thrm}[\cite{Zyczkowski2002}]\label{th:HSVolume}
It holds that
\begin{eqnarray}
\vol_{\rH\rS} (\rD(\bbC^N)) =
\sqrt{N}(2\pi)^{\binom{N}{2}}\frac{\prod^N_{k=1}\Gamma(k)}{\Gamma(N^2)}.
\end{eqnarray}
\end{thrm}

\begin{proof}
The infinitesimal distance takes a particularly simple form $\dif
s^2_{\rH\rS}=\Inner{\dif\rho}{\dif\rho}_{\rH\rS}$ valid for any
dimension $N$. Making use of the diagonal form
$\rho=\bsU\Lambda\bsU^\dagger$, we may write
$$
\dif\rho = \bsV\Pa{\dif\Lambda +
[\bsV^\dagger\dif\bsV,\Lambda]}\bsV^\dagger,
$$
where $\bsV\in\sfU(N)/\bbT^N$. Thus the differential of arc length
can be rewritten as
\begin{eqnarray*}
\dif s^2_{\rH\rS} = \sum^N_{j=1}\dif\lambda^2_j +
2\sum^N_{1\leqslant i<j\leqslant
N}(\lambda_i-\lambda_j)^2\abs{\Innerm{i}{\bsV^\dagger\dif\bsV}{j}}^2.
\end{eqnarray*}
Apparently, $\sum^N_{j=1}\dif\lambda_j=0$ since
$\sum^N_{j=1}\lambda_j=1$. Thus
$$
\dif s^2_{\rH\rS} = \sum^{N-1}_{i,j=1}\dif\lambda_i
g_{ij}\dif\lambda_j + 2\sum_{1\leqslant i<j\leqslant
N}(\lambda_i-\lambda_j)^2\abs{\Innerm{i}{\bsV^\dagger\dif\bsV}{j}}^2.
$$
The HS volume element is given by
$$
\dif V_{\rH\rS} =
\sqrt{N}2^{\binom{N}{2}}V_N(\lambda)^2\prod^{N-1}_{k=1}\dif\lambda_k
\prod_{1\leqslant i<j\leqslant
N}\re(\bsV^\dagger\dif\bsV)_{ij}\im(\bsV^\dagger\dif\bsV)_{ij}.
$$
The corresponding volume element gains a factor
$\sqrt{\det(g)}=\sqrt{N}$, where $g=(g_{ij})_{(N-1)\times(N-1)}$ is
the metric in the $(N-1)$-dimensional simplex $\Delta_{N-1}$ of
eigenvalues. Note that $g=\I_{N-1}+\proj{\bse}$ for
$(N-1)$-dimensional vector $\bse:=(1,\ldots,1)^\t$. Therefore
\begin{eqnarray*}
\vol_{\rH\rS}(\rD(\bbC^N)) &=& \int_{\rD(\bbC^N)}\dif
V_{\rH\rS}=\int_{\Delta_{N-1}\cap\cC_N}V_N(\lambda)^2\sqrt{N}\prod^{N-1}_{k=1}\dif\lambda_k
\times\int_{\sfU(N)/\bbT^N}2^{\binom{N}{2}}[\bsV^\dagger\dif\bsV]\\
&=&\vol_{\rH\rS}(\sfU(N)/\bbT^N)
\int_{\Delta_{N-1}\cap\cC_N}V_N(\lambda)^2\sqrt{N}\prod^{N-1}_{k=1}\dif\lambda_k,
\end{eqnarray*}
where
\begin{eqnarray*}
\int_{\Delta_{N-1}\cap\cC_N}V_N(\lambda)^2\prod^{N-1}_{k=1}\dif\lambda_k
=\frac1{N!}\int_{\bbR^N_{\geqslant0}}\delta\Pa{1-\sum^N_{k=1}\lambda_k}V_N(\lambda)^2\prod^N_{k=1}\dif\lambda_k
=\frac{\Pa{\prod^N_{k=1}\Gamma(k)}^2}{\Gamma(N^2)}.
\end{eqnarray*}
Substituting the last result into the above, we get the desired
result.
\end{proof}

\begin{remark}
We can partition the set $\rD(\bbC^N)$ of all density matrices
acting on $\bbC^N$ according to the spectrum
$\Lambda=(\lambda_1,\ldots,\lambda_N)\in\Delta_{N-1}\cap\cC_N$:
\begin{eqnarray}
\rD(\bbC^N) =
\biguplus_{\Lambda\in\Delta_{N-1}\cap\cC_N}\cU_{\Lambda}.
\end{eqnarray}
From this observation and previous discussion, we get that
\begin{eqnarray*}
&&\vol_{\rH\rS}(\rD(\bbC^N))=\sqrt{N}\frac{\Pa{\prod^N_{k=1}\Gamma(k)}^2}{\Gamma(N^2)}\cdot\vol_{\rH\rS}(\sfU(N)/\bbT^N)\\
&&=\sqrt{N}\int_{\Delta_{N-1}\cap\cC_N}V_N(\Lambda)^2[\dif\Lambda]\cdot
\vol_{\rH\rS}(\sfU(N)/\bbT^N)\\
&&= \sqrt{N}\int_{\Delta_{N-1}\cap\cC_N}
\vol_{\rH\rS}(\cU_{\Lambda})[\dif\Lambda].
\end{eqnarray*}
Based on it, we infer that every adjoint orbit $\cU_\Lambda$ is
subject to the distrubtion
\begin{eqnarray}\label{eq:mlambda}
\dif m(\Lambda)=\sqrt{N}\delta(1-\Tr{\Lambda})[\dif\Lambda].
\end{eqnarray}
It holds that
\begin{eqnarray}\label{eq:volHSagain}
\vol_{\rH\rS}(\rD(\bbC^N)) = \int_{\cC_N\cap\Delta_{N-1}}
\vol_{\rH\rS}(\cU_{\Lambda})\dif m(\Lambda).
\end{eqnarray}
\end{remark}

\section{Symplectic volumes of regular co-adjoint orbits}\label{sect:4}

Consider a compact connected Lie group $K$ with its Lie algebra
$\k$. Let $T$ be the maximal torus of $K$ with its Lie algebra
$\liet$, namely, the Cartan subalgebra. Choose an $K$-invariant
inner product on $\k$. By restricting such $K$-invariant inner
product on $\liet$, we can identify $\liet^*$, the dual space of
$\liet$, with $\liet$.

Let $R^+$ be the set of all positive roots of $K$. Then $R^+$
determines the positive Weyl chamber $\liet_{\geqslant0}$ via the
way
\begin{eqnarray}
\mathrm{i}\liet_{\geqslant0}=\Set{\bsX\in\mathrm{i}\liet:
\alpha(\bsX)\geqslant0,\forall \alpha\in R^+}.
\end{eqnarray}
Analogously, we also identity $\liet^*_{\geqslant0}$ with
$\liet_{\geqslant0}$.
%
%

Consider an element $\hat\lambda\in\mathrm{i}\liet^*_{>0}$, the
interior of positive Weyl chamber $\mathrm{i}\liet^*_{\geqslant0}$.
Let
\begin{eqnarray}
\cO_{\hat\lambda} :=K\cdot{\hat\lambda}=\Set{\Ad^*_g\hat\lambda:
g\in K}
\end{eqnarray}
be the co-adjoint orbit, identified by co-adjoint action of $K$ on
$\mathrm{i}\liet^*_{>0}\cong\mathrm{i}\liet_{>0}$. There is a
$K$-equivariantly diffeomorphism
$$
\cO_{\hat\lambda} \cong K/T.
$$
Moreover, the coadjoint orbit $\cO_{\hat\lambda}$ carry the
so-called standard symplectic form, i.e., Kirillov-Kostant-Souriau
form $\omega_{\rK\rK\rS}$ \cite{Silva2001}. Now we get a symplectic
form $(\cO_{\hat\lambda},\omega_{\rK\rK\rS})$. The top degree of
$\omega_{\rK\rK\rS}$ will identify a volume form
$\frac{\omega^d_{\rK\rK\rS}}{d!}$, where
$d:=\frac12\dim(\cO_{\hat\lambda}) = \frac12\dim(K/T)$.
\begin{definition}
The Liouville measure on $\cO_{\hat\lambda}$ is defined by
\begin{eqnarray}
\mu_{\cO_{\hat\lambda}}(B):=\int_B\frac{\omega^d_{\rK\rK\rS}}{d!},
\end{eqnarray}
where $B$ is a Borel subset of $\cO_{\hat\lambda}$. The symplectic
volume of $\cO_{\hat\lambda}$ is just
$\vol_{\mathrm{symp}}(\cO_{\hat\lambda}) =
\mu_{\cO_{\hat\lambda}}(\cO_{\hat\lambda})$.
\end{definition}

\begin{prop}[\cite{Berline1992}]
Let $T$ be a maximal torus of a compact connected Lie group $K$. Let
$\hat\lambda\in\mathrm{i}\liet^*_{>0}$ for which $\cO_{\hat\lambda}$
is the coadjoint orbit through $\hat\lambda$. The symplectic volume
of $\cO_{\hat\lambda}$, with respect to the KKS form
$\omega_{\rK\rK\rS}$, is given by
\begin{eqnarray}\label{eq:sympvol}
\vol_{\mathrm{symp}}(\cO_{\hat\lambda}) = \Pa{\prod_{\alpha\in
R^+}\Inner{{\hat\lambda}}{\alpha}}\Pa{\prod_{\alpha\in
R^+}\frac{2\pi}{\Inner{\rho}{\alpha}}},
\end{eqnarray}
where $\rho:=\frac12\sum_{\alpha\in R^+}\alpha$.
\end{prop}

\begin{remark}
According to Harish-Chandra's volume formula \cite{Harish1975},
\begin{eqnarray}
\vol_{\rH\rS}(K/T)=\prod_{\alpha\in
R^+}\frac{2\pi}{\Inner{\rho}{\alpha}}.
\end{eqnarray}
Based on this formula, we get that
\begin{eqnarray}
\vol_{\mathrm{symp}}(\cO_{\hat\lambda}) = \Pa{\prod_{\alpha\in
R^+}\Inner{{\hat\lambda}}{\alpha}}\vol_{\rH\rS}(K/T).
\end{eqnarray}
\end{remark}

\begin{exam}
Let $K=\sfU(N)$. The maximal torus of $\sfU(N)$ is $T=\bbT^N$, and
its Lie algebra $\liet$ such that
$\mathrm{i}\liet^*\cong\mathrm{i}\liet\cong\bbR^N$. Moreover
$\mathrm{i}\liet^*_{>0}\cong \cC_N$. Then the set of all positive
roots $R^+$, where
\begin{eqnarray}
R^+=\Set{\alpha_i-\alpha_j:,1\leqslant i<j\leqslant N},
\end{eqnarray}
where $\alpha_k$ is a linear functional taking the $k$-th position's
imaginary part. Moreover $d=\abs{R^+}=\binom{N}{2}$. The symplectic
volume of the coadjoint orbit $\cO_{\hat\lambda}$ is
\begin{eqnarray}\label{eq:symvolumeoforbit}
\vol_{\mathrm{symp}}(\cO_{\hat\lambda}) =
V_N(\hat\lambda)\vol_{\rH\rS}(\sfU(N)/\bbT^N).
\end{eqnarray}
\end{exam}

When we consider both HS volume and symplectic volume for the same
manifold $\cO_{\hat\lambda}$, we will get the following relationship
between such volumes.
\begin{prop}\label{prop:relationHSSymp}
Let ${\hat\lambda}\in
\mathrm{i}\liet^*_{>0}\cong\mathrm{i}\liet_{>0}$ for the unitary
group $\sfU(N)$. The HS volume and the symplectic volume of
$\cO_{\hat\lambda}$ are connected via the formula:
\begin{eqnarray}\label{eq:HSvsSymp}
\vol_{\rH\rS}(\cO_{\hat\lambda})=
V_N(\hat\lambda)\times\vol_{\mathrm{symp}}(\cO_{\hat\lambda}).
\end{eqnarray}
\end{prop}

\begin{proof}
According to Proposition~\ref{prop:HSvolumeoforbit},
\begin{eqnarray*}
\vol_{\rH\rS}(\cO_{\hat\lambda}) &=&
V_N(\hat\lambda)^2\vol_{\rH\rS}(\sfU(N)/\bbT^N)=
V_N(\hat\lambda)\Br{V_N(\hat\lambda)\vol_{\rH\rS}(\sfU(N)/\bbT^N)}\\
&=& V_N(\hat\lambda)\vol_{\mathrm{symp}}(\cO_{\hat\lambda}),
\end{eqnarray*}
where the last equality is obtained from
Eq.~\eqref{eq:symvolumeoforbit}.
\end{proof}

\section{Duistermaat-Heckman measure}\label{sect:5}

The Duistermaat-Heckman ($\dh$) Theorem, introduced in
\cite{DH1982}, revolutionized symplectic geometry and mathematical
physics by uncovering profound connections between Hamiltonian
dynamics, localization, and topology. Its key conclusions ---
including the exact stationary phase approximation, the convexity of
images of moment map, the polynomiality of reduced symplectic
volumes, relations to equivariant cohomology and localization
\cite{Bytsenko2005}, non-Abelian generalizations, applications in
Physics, and combinatorial and algebraic consequences --- have each
had significant impact.

Central to these developments is the powerful Duistermaat-Heckman
($\dh$) measure \cite{Zhang2019}. This measure provides a crucial
tool across symplectic geometry, representation theory, and
mathematical physics. Its strength lies in transforming complex
global geometric problems into manageable local computations
concentrated at fixed points, thereby effectively bridging
symplectic geometry with representation theory, combinatorics, and
physics.

\begin{definition}[Push-forward of a measure]
Given measurable spaces $(X,\cF)$ and $(Y,\cG)$, a measurable
mapping $\Phi:X\to Y$ and a measure $\mu:\cF\to[0,+\infty)$, the
push-forward of $\mu$ is defined to be the measure
$\Phi_*\mu:\cG\to[0,+\infty)$ given by
\begin{eqnarray*}
(\Phi_*\mu)(B) := \mu\Pa{\Phi^{-1}(B)}\quad\text{for}\quad B\in\cG.
\end{eqnarray*}
\end{definition}

\begin{prop}[Change of variables formula]
A measurable function $f$ on $Y$ is integrable with respect to the
push-forward measure $\Phi_*\mu$ if and only if the composition
$\Phi^*f=f\circ \Phi$ is integrable with respect to the measure
$\mu$. In that case, the integrals coincide
\begin{eqnarray*}
\int_{Y}f\dif (\Phi_*\mu) = \int_{X}f\circ \Phi\dif\mu =
\int_{X}\Phi^*f\dif\mu.
\end{eqnarray*}
In the sense of distribution,
$\Inner{\Phi_*\mu}{f}=\Inner{\mu}{\Phi^*f}$.
\end{prop}

\begin{definition}[Liouville measure]
Let $(\bsM,\omega)$ be a symplectic manifold of $2n$ dimension. A
Borel set in $\bsM$ is a set generated from compact subsets of
$\bsM$ under countable union and complementation. Given a Borel set
$B$ in $\bsM$, the \emph{Liouville measure} of $B$ is defined as
\begin{eqnarray*}
\mu(B) = \int_B \frac{\omega^n}{n!}.
\end{eqnarray*}
The \emph{symplectic volume} of $\bsM$ is defined by
\begin{eqnarray*}
\vol_{\mathrm{symp}}(\bsM)=\mu(\bsM) = \int_{\bsM}
\frac{\omega^n}{n!}.
\end{eqnarray*}
\end{definition}


In order to define Duistermaat-Heckman measure, we need the
following notion of moment map (or momentum map) is a fundamental
concept in symplectic geometry and mathematical physics, providing a
bridge between symplectic group actions and conserved quantities. It
generalizes the idea of conserved momenta in Hamiltonian mechanics
to more abstract geometric settings.
\begin{definition}[Moment map]
Given a smooth manifold $\bsM$ equipped with a symplectic form
$\omega$ (a closed, non-degenerate $2$-form). A Lie group $K$ acts
on $\bsM$ via symplectomorphisms (symmetry transformations
preserving $\omega$). The so-called \emph{moment map} is a map
$\Phi:\bsM\to\mathrm{i}\k^*$ satisfying:
\begin{enumerate}[(i)]
\item For every $X\in\k$, the function
$\Phi^X(m):=\Inner{\Phi(m)}{X}$ is a Hamiltonian function for the
vector field $X_{\bsM}$ on $\bsM$ generated by $X$:
$$
\dif\Phi^X = \iota_{X_{\bsM}}\omega.
$$
\item $\Phi$ is $K$-equivariant with respect to the coadjoint action
$\Ad^*$ on $\mathrm{i}\k^*$:
$$
\Phi(g\cdot m)=(\Ad^*_g\Phi)(m),\quad \forall g\in K.
$$
\end{enumerate}
In such a case, this action of $K$ on $\bsM$ is called
\emph{Hamiltonian action}. At this time, $(\bsM,K,\omega)$ is called
\emph{Hamiltonian $K$-manifold}.
\end{definition}

\begin{definition}[Duistermaat-Heckman measure]
Let $(\bsM,K,\omega)$ (here $\omega$ is a symplectic form) be a
compact, connected Hamiltonian $K$-manifold of dimension $2n$ and a
choice of moment map $\Phi:\bsM\to\mathrm{i}\k^*$.
\begin{enumerate}[(i)]
\item The \emph{non-Abelian}
Duistermaat-Heckman measure $\dh^K_{\bsM}$ is defined as follows:
\begin{eqnarray}
\dh^K_{\bsM} = \frac1{p_K}(\tau_K)_*\Phi_*\mu
\end{eqnarray}
where $\tau_K:\mathrm{i}\k^*\to\mathrm{i}\liet^*_{\geqslant0}$ is
defined as $\tau_K(\cO_{\hat\lambda})=\set{\hat\lambda}$ for
$\hat\lambda\in \mathrm{i}\liet^*_{\geqslant0}$ and
$p_K(\hat\lambda)=\vol_{\mathrm{symp}}(\cO_{\hat\lambda})$, which is
known from Eq.~\eqref{eq:sympvol}.
\item The \emph{Abelian}
Duistermaat-Heckman measure is defined as:
\begin{eqnarray}
\dh^T_{\bsM} = \tau_*\Phi_*\mu,
\end{eqnarray}
where $\tau:\mathrm{i}\k^*\to\mathrm{i}\liet^*$ is the projection
dual to the inclusion map $\liet\hookrightarrow\k$.
\end{enumerate}
\end{definition}

The computation of the density of the pushforward measure with
respect to the Lebesgue measure on the positive Weyl chamber is
governed by the Duistermaat-Heckman (DH) theorem \cite{DH1982} and
its generalizations. But in a special case, for instance, where the
underlying symplectic manifold is just a coadjoint orbit through a
regular element in $\mathrm{i}\liet^*_{>0}$, there is an explicit
formula for computing the density. In the next subsection, we
describe it explicitly.

\subsection{Harish-Chandra formula and derivative principle}

The following result gives the Fourier transform formula of Abelian
Duistermaat-Heckman measure for the co-adjoint orbit
$\cO_{\hat\lambda}$. By this, we can identify the analytical
expression of Abelian Duistermaat-Heckman measure.

\begin{thrm}[Harish-Chandra, \cite{Harish1975}]
The Fourier transform of abelian Duistermaat-Heckman
$\dh^T_{\cO_{\hat\lambda}}$ measure is given by
\begin{eqnarray}
\Inner{\dh^T_{\cO_{\hat\lambda}}}{e^{\mathrm{i}\Inner{-}{X}}} =
\sum_{w\in
W}(-1)^{l(w)}e^{\mathrm{i}\Inner{w\hat\lambda}{X}}\prod_{\alpha>0}\frac1{\mathrm{i}\Inner{\alpha}{X}}.
\end{eqnarray}
for every $X$ in the Lie algebra of $T$ which is not orthogonal to
any root. That is,
\begin{eqnarray}
\dh^T_{\cO_{\hat\lambda}} = \sum_{w\in
W}(-1)^{l(w)}\delta_{w\hat\lambda}\star
H_{-\alpha_1}\star\cdots\star H_{-\alpha_R}.
\end{eqnarray}
Here $l(w)$ is the length of the Weyl group element $w\in W$, and
$\delta_\alpha$ for the Dirac measure at $\alpha$; $\star$ means the
convolution.
\end{thrm}

Subsequently, the so-called derivative principle is provided in
which the non-Abelian Duistermaat-Heckman measure is obtained from
Abelian Duistermaat-Heckman measure.
\begin{thrm}[Derivative principle, \cite{Christandl2014}]
It holds that
\begin{eqnarray*}
\dh^K_{\cO_{\hat\lambda}}\Big|_{\mathrm{i}\liet^*_{\geqslant0}} =
\Pa{\prod_{\alpha>0}\partial_{-\alpha}}\dh^T_{\cO_{\hat\lambda}}\Big|_{\mathrm{i}\liet^*_{\geqslant0}}.
\end{eqnarray*}
\end{thrm}

Consider more generally the action of $\tilde T$ on a co-adjoint
$K$-orbit $\cO_{{\hat\lambda}}$ induced by a group homomorphism
$\tilde T\to T\subset K$. Note that it follows directly that
$$
\dh^{\tilde T}_{\cO_{\hat\lambda}}=\pi_*\dh^T_{\cO_{\hat\lambda}}
$$
where $\pi:\mathrm{i}\liet^*\to\mathrm{i}\tilde\liet^*$ is the
restriction map from the dual Lie algebra $\liet^*$ of $T$ to that
$\tilde\liet^*$ of $\tilde T$. Thus we get that
$$
\dh^{\tilde T}_{\cO_{\hat\lambda}} = \sum_{w\in
W}(-1)^{l(w)}\delta_{\pi(w\hat\lambda)}\star
H_{-\pi(\alpha_1)}\star\cdots\star H_{-\pi(\alpha_R)}.
$$

In what follows, we will now use the non-Abelian Heckman algorithm
to treat the case of random two-qubit states with fixed,
non-degenerate global eigenvalue spectrum. That is, we consider the
action of $\tilde K=\SU(2)\times \SU(2)$ on a coadjoint
$K=\SU(4)$-orbit through a point $\hat\lambda$ in the interior of
positive Weyl chamber of $K$.

\begin{itemize}
\item Let $\tilde T$ be the maximal torus of $\tilde K$.
\item Symmetric group $W=S_4$ is the Weyl group of $K=\SU(4)$.
\item $\dh^{\tilde T}_{\cO_{\hat\lambda}}=\sum_{w\in S_4}\sign(w)\delta_{\pi(w\hat\lambda)}\star H_{-\pi(\alpha_1)}\star\cdots\star
H_{-\pi(\alpha_6)}$, where $\set{\alpha_k:k=1,\ldots,6}$ is the set
of positive roots of $K=\SU(4)$.
\item The dual Lie algebra $\tilde\liet^*$ of $\tilde T$ can be identified with
$\mathrm{i}\tilde\liet^*\cong\bbR^2$.
\item The maps $\pi$ is given by
$$
\pi (\hat\lambda_1,\ldots,\hat\lambda_4)=
2\mathrm{i}(\hat\lambda_1+\hat\lambda_2,\hat\lambda_1+\hat\lambda_3).
$$
\item One computes readily that the $-\pi(\alpha_k)$ are precisely the
weights
$$
\set{(-2,2),(-2,0)_{(2)},(-2,-2),(0,-2)_{(2)}}.
$$
In particular, two negative roots of $\tilde K$ are contained in
this list (each of them is in fact contained twice).
\end{itemize}

The following result appears in \cite{Christandl2014} without proof.
We provide a detailed proof for completeness.
\begin{prop}
The non-Abelian Duistermaat-Heckman measure for the action of
$\tilde K=\SU(2)\times \SU(2)$ on a coadjoint $K$-orbit
$\cO_{\hat\lambda}=K\cdot\hat\lambda$ with $\hat\lambda$, where
$K=\SU(4)$, in the interior of positive Weyl chamber of $T$ is given
by
\begin{eqnarray}
\dh^{\tilde
K}_{\cO_{\hat\lambda}}\Big|_{\mathrm{i}\tilde\liet^*_{\geqslant0}} =
\Pa{\sum_{w\in S_4}\sign(w)\delta_{\pi(w\lambda)}}\star
H_{(-2,2)}\star H_{(-2,0)}\star H_{(-2,-2)}\star
H_{(0,-2)}\Big|_{\mathrm{i}\tilde\liet^*_{\geqslant0}},
\end{eqnarray}
where $\mathrm{i}\tilde\liet^*_{\geqslant0}\cong\bbR^2_{\geqslant0}$
is the positive Weyl chamber of $\tilde T$, the maximal torus of
$\tilde K$.
\end{prop}

\begin{proof}
Recall that the Weyl group of $K=\SU(4)$ is the symmetric group
$W=S_4$, with $(-1)^{l(w)}$ equal to the signum of a permutation
$w\in S_4$. Then
$$
\dh^{\tilde T}_{\cO_{\hat\lambda}} = \pi_*\dh^T_{\cO_{\hat\lambda}}
= \sum_{w\in S_4}\sign(w)\delta_{\pi(w{\hat\lambda})}\star
H_{-\pi(\alpha_1)}\star\cdots\star H_{-\pi(\alpha_6)},
$$
where $T$ is is the maximal torus of $K=\SU(4)$, and $\tilde T$ is
the maximal torus of $\tilde K=\SU(2)\times\SU(2)$, and
$\set{\alpha_k:k=1,\ldots,6}$ are the positive roots of $K=\SU(4)$,
i.e.,
\begin{eqnarray*}
\alpha_1=(1,-1,0,0),\alpha_2=(1,0,-1,0),\alpha_3=(1,0,0,-1),\\
\alpha_4=(0,1,-1,0),\alpha_5=(0,1,0,-1),\alpha_6=(0,0,1,-1).
\end{eqnarray*}
Since $\hat\lambda_1>\cdots>\hat\lambda_4$, it follows that this
$\pi$ maps each $\hat\lambda$ to
$2\mathrm{i}(\hat\lambda_1+\hat\lambda_2,\hat\lambda_1+\hat\lambda_3)$,
where $2(\hat\lambda_1+\hat\lambda_2)$ and
$2(\hat\lambda_1+\hat\lambda_3)$ are the maximal eigenvalues of two
reduced density matrices determined by a global density matrix
$\Lambda$, where
$\Lambda=\diag(\hat\lambda_1+\frac14,\ldots,\hat\lambda_4+\frac14)$.
One computes readily that the $-\pi(\alpha_k)$ are precisely the
weights
$$
\set{(-2,2),(-2,0)_{(2)},(-2,-2),(0,-2)_{(2)}}.
$$
Indeed,
\begin{eqnarray*}
-\pi(\alpha_1)=-\pi(\alpha_6)=(0,-2),-\pi(\alpha_3)=(-2,-2),\\
-\pi(\alpha_2)=-\pi(\alpha_5)=(-2,0),-\pi(\alpha_4)=(-2,2).
\end{eqnarray*}
That is,
$$
\set{-\pi(\alpha_k):k\in[6]}=\set{(-2,2),(-2,0)_{(2)},(-2,-2),(0,-2)_{(2)}}.
$$
In particular, the two negative roots $\set{(-2,0),(0,-2)}$ of
$\tilde K=\SU(2)\times\SU(2)$ are contained in this list t (each of
them is in fact contained twice). From the above discussion, we see
that
\begin{eqnarray*}
\dh^{\tilde T}_{\cO_{\hat\lambda}} = \Pa{\sum_{w\in
S_4}\sign(w)\delta_{\pi(w\hat\lambda)}}\star H_{(0,-2)}\star
H_{(-2,0)}\star H_{(-2,-2)}\star H_{(-2,2)}\star H_{(-2,0)}\star
H_{(0,-2)}.
\end{eqnarray*}
Therefore, we arrived at the following formula:
\begin{eqnarray*}
&&\dh^{\tilde
K}_{\cO_{\hat\lambda}}\Big|_{\mathrm{i}\tilde\liet^*_{\geqslant0}} =
\Pa{\prod_{\alpha>0}\partial_{-\alpha}}\dh^{\tilde
T}_{\cO_{\hat\lambda}}\Big|_{\mathrm{i}\tilde\liet^*_{\geqslant0}}\\
&&=\Pa{\sum_{w\in S_4}\sign(w)\delta_{\pi(w{\hat\lambda})}}\star
H_{(-2,-2)}\star H_{(-2,2)}\star H_{(-2,0)}\star
H_{(0,-2)}\Big|_{\mathrm{i}\tilde\liet^*_{\geqslant0}}.
\end{eqnarray*}
We are done.
\end{proof}

We have already known the fact that a density matrix of a generic
two-qubit state is represented by the $2\times 2$ block matrix
$$
\rho_{12}=\Pa{\begin{array}{cc}
           \bsA & \bsC \\
           \bsC^\dagger & \bsB
         \end{array}
},
$$
where $\bsA,\bsB$ are $2\times 2$ positive semi-definite complex
matrices. Now we present the two reduced marginal states,
respectively:
$$
\rho_1=\Pa{\begin{array}{cc}
             \Tr{\bsA} & \Tr{\bsC} \\
             \Tr{\bsC^\dagger} & \Tr{\bsB}
           \end{array}
},\quad \rho_2=\bsA+\bsB.
$$
From this, we see that for $\rho_{12}=\Lambda$ with
$\Lambda_1\geqslant\cdots\geqslant \Lambda_4\geqslant0$ and
$\sum^4_{j=1}\Lambda_j=1$, it can be rewritten as $\rho-\frac{\I_4}4
= \diag(\hat\lambda_1,\ldots,\hat\lambda_4)$, where
$\hat\lambda_1\geqslant\cdots\geqslant\hat\lambda_4$ and
$\sum^4_{j=1}\hat\lambda_j=0$. Thus $\hat\lambda$ is in the positive
Weyl chamber $\liet_{\geqslant0}$ of $T$.

Now the map $\rho\mapsto (\rho_1,\rho_2)$, where
$\rho_1=\Ptr{2}{\rho}$ and $\rho_2=\Ptr{1}{\rho}$, can be
equivalently represented as
\begin{eqnarray*}
\pi(\hat\lambda)&=&\Pa{\Pa{\begin{array}{cc}
                       \hat\lambda_1+\hat\lambda_2 & 0 \\
                       0 & \hat\lambda_3+\hat\lambda_4
                     \end{array}
},\Pa{\begin{array}{cc}
                       \hat\lambda_1+\hat\lambda_3 & 0 \\
                       0 & \hat\lambda_2+\hat\lambda_4
                     \end{array}
}}\\
&\defeq&
(2(\hat\lambda_1+\hat\lambda_2),2(\hat\lambda_1+\hat\lambda_3)),
\end{eqnarray*}
where
$$
\Pa{\begin{array}{cc}
                       \hat\lambda_1+\hat\lambda_2 & 0 \\
                       0 & \hat\lambda_3+\hat\lambda_4
                     \end{array}
}\defeq \hat\lambda_1+\hat\lambda_2 -
(\hat\lambda_3+\hat\lambda_4)=2(\hat\lambda_1+\hat\lambda_2)
$$
and
$$
\Pa{\begin{array}{cc}
                       \hat\lambda_1+\hat\lambda_3 & 0 \\
                       0 & \hat\lambda_2+\hat\lambda_4
                     \end{array}
}\defeq \hat\lambda_1+\hat\lambda_3 -
(\hat\lambda_2+\hat\lambda_4)=2(\hat\lambda_1+\hat\lambda_3).
$$

Denote $\bsh=\diag(1,-1)$. Then
$\diag(\hat\lambda_1+\hat\lambda_2,\hat\lambda_3+\hat\lambda_4)=(\hat\lambda_1+\hat\lambda_2)\cdot\bsh$
and
$\diag(\hat\lambda_1+\hat\lambda_3,\hat\lambda_2+\hat\lambda_4)=(\hat\lambda_1+\hat\lambda_3)\cdot\bsh$.
Since $\hat\lambda_1+\hat\lambda_3\geqslant
\hat\lambda_2+\hat\lambda_4$ and
$(\hat\lambda_1+\hat\lambda_3)+(\hat\lambda_2+\hat\lambda_4)=0$,
thus $\hat\lambda_1+\hat\lambda_3\geqslant0$, a fortiori
$\hat\lambda_1+\hat\lambda_2\geqslant
\hat\lambda_1+\hat\lambda_3\geqslant0$. Note that
$\hat\lambda_1+\hat\lambda_4$ is not non-negative in general.

\subsection{Boysal-Vergne-Paradan jump formula}

In 2009, Bosyal and Vergne published a paper \cite{Boysal2009} in
which they investigated the push-forward of Lebesgue measures on the
cone $\bbR^N_{\geqslant0}$ along a linear map. Later in his PhD
thesis, Walter found that the density of Abelian Duistermaat-Heckman
measure with respect to the Lebesgue measure on the affine hull of
the Abelian moment polytope is given by the volume of a parametrized
polytope \cite{Walter2014}. In view of this result, Boysal and
Vergne's result is connected with Abelian Duistermaat-Heckman
measure, and can be adapted to calculate the density. To this end,
we need firstly introduce some notions involved in question.

\begin{definition}[\cite{Boysal2009}]
Let $\bse\in V$ be a primitive vector. It defines a hyperplane in
$V^*$ (the dual space of $V$):
$$
W = \Set{\alpha\in V^*:\Inner{\alpha}{\bse}=0}\defeq \bse^\perp.
$$
Let $P$ be a polynomial function on $V^*$ and let $\Psi$ be a
sequence of vectors not belong to $W$, i.e., $\Psi\cap W=\emptyset$.
We define, for $\alpha\in V^*$,
\begin{eqnarray}
\Pol(P,\Psi,\bse)(\alpha) =
\Res_{z=0}\Br{\Pa{P(\partial_{\bsx})\frac{e^{\Inner{\alpha}{\bsx+z\bse}}}{\prod_{\psi\in\Psi}\Inner{\psi}{\bsx+z\bse}}}_{\bsx=\zero}}.
\end{eqnarray}
\end{definition}

\begin{remark}
Note that  the function $\Pol(P,\Psi,\bse)$ depends only on the
restriction $p$ of $P$ to $W$. If $p$ is a polynomial function on
$W=\bse^\perp$. We define
$$
\Pol(p,\Psi,\bse):=\Pol(P,\Psi,\bse),
$$
where $P$ is any polynomial on $V^*$ extending $p$.
\end{remark}

\begin{definition}[Wall]
Let $\Psi=[\psi_1,\ldots,\psi_N]$ be a sequence of non-zero, not
necessarily distinct, linear forms on $V$, i.e., $\psi_k\in V^*$,
lying in an open half-space. If $\Psi$ spans the whole space $V^*$
with $\dim(V)=n$, then a \emph{wall} of $\Psi$ is a (real)
hyperplane generated by $n-1$ linearly independent elements of
$\Psi$.
\end{definition}

Note that $V^*$ is separated into two open half-spaces by the wall
$W:=\set{\phi\in V^*:\Inner{\phi}{\bse}=0}$ in $V^*$. Let $V^*_\pm$
denote the corresponding open half-spaces, that is,
\begin{eqnarray*}
V^*_+:=\Set{\phi\in V^*:\Inner{\phi}{\bse}>0}\text{ and
}V^*_-:=\Set{\phi\in V^*:\Inner{\phi}{\bse}<0}.
\end{eqnarray*}
Let $C_1\subset V^*_+$ and $C_2\subset V^*_-$ be two chambers on two
sides of $W$ and adjacent. We choose the measures $\dif \bsx$ on $V$
and $\dif\phi$ on $V^*$, respectively, and choose Lebesgue measure
$\dif\bst$ on $\bbR^N$. We also choose the measure $\dif w$ on $W$
such that $\dif\phi =\dif w\dif t$ with $t=\Inner{\phi}{\bse}$ for
$\phi\in V^*$. Based on this, we can write $\Psi$ as
$$
\Psi =[\Psi_0,\Psi^+,\Psi^-],
$$
where $\Psi^\pm =\Psi\cap V^*_\pm$, and $\Psi_0=\Psi\cap W$.

\begin{thrm}[Boysal-Vergne-Paradan jump formula, \cite{Boysal2009}]
Let $W$ be a wall, determined by a vector $\bse\in V$, $\Psi$ be a
sequence of vectors spanning the whole space $V^*$. Denote
$\Psi_0:=\Psi\cap W$. Let $v_{12}=v(\Psi_0,\dif w,C_{12})$ be the
polynomial function on $W$ associated to the chamber $C_{12}$ of
$\Psi_0$. Then, if $\Inner{C_1}{\bse}>0$ and $V_{12}$ is any
extension of $v_{12}$ on $V^*$,
\begin{eqnarray}
v(\Psi,\dif \bsx,C_1) - v(\Psi,\dif \bsx,C_2) =
\Pol(v_{12},\Psi\backslash\Psi_0,\bse).
\end{eqnarray}
\end{thrm}

Although the following result is obtained in \cite{Christandl2014},
the details of proof is not provided. For reader's convenience, we
present it here using Boysal-Vergne-Paradan jump formula.
\begin{prop}\label{prop:4density}
The measure $H_{(-2,2)}\star H_{(-2,0)}\star H_{(-2,-2)}\star
H_{(0,-2)}$ has Lebesgue density
\begin{eqnarray*}
p(r,s) =
\begin{cases}
p_1(r,s)\equiv\frac{(r+s)^2}{64},&\text{if }(r,s)\in C_1=\set{(r,s):0\leqslant s\leqslant -r};\\
p_2(r,s)\equiv\frac{r^2+2rs-s^2}{64},&\text{if }(r,s)\in C_2=\set{(r,s):r\leqslant s\leqslant0};\\
p_3(r,s)\equiv\frac{r^2}{32},&\text{if }(r,s)\in C_3=\set{(r,s): s\leqslant r\leqslant0},\\
p_0(r,s)\equiv0,&\text{if }(r,s)\in C_0=\bbR^2\backslash (C_1\cup
C_2\cup C_3).
\end{cases}
\end{eqnarray*}
\end{prop}
The density function $p(r,s)$ (its graph has already appeared in
\cite{Christandl2014}) and its support that is decomposed into three
chammbers can be visualized in the following
Figure~\ref{fig:chamber-support}.
\begin{figure}[ht]\centering
\subfigure[Chambers and support] {\begin{minipage}[b]{0.5\linewidth}
\includegraphics[width=.7\textwidth]{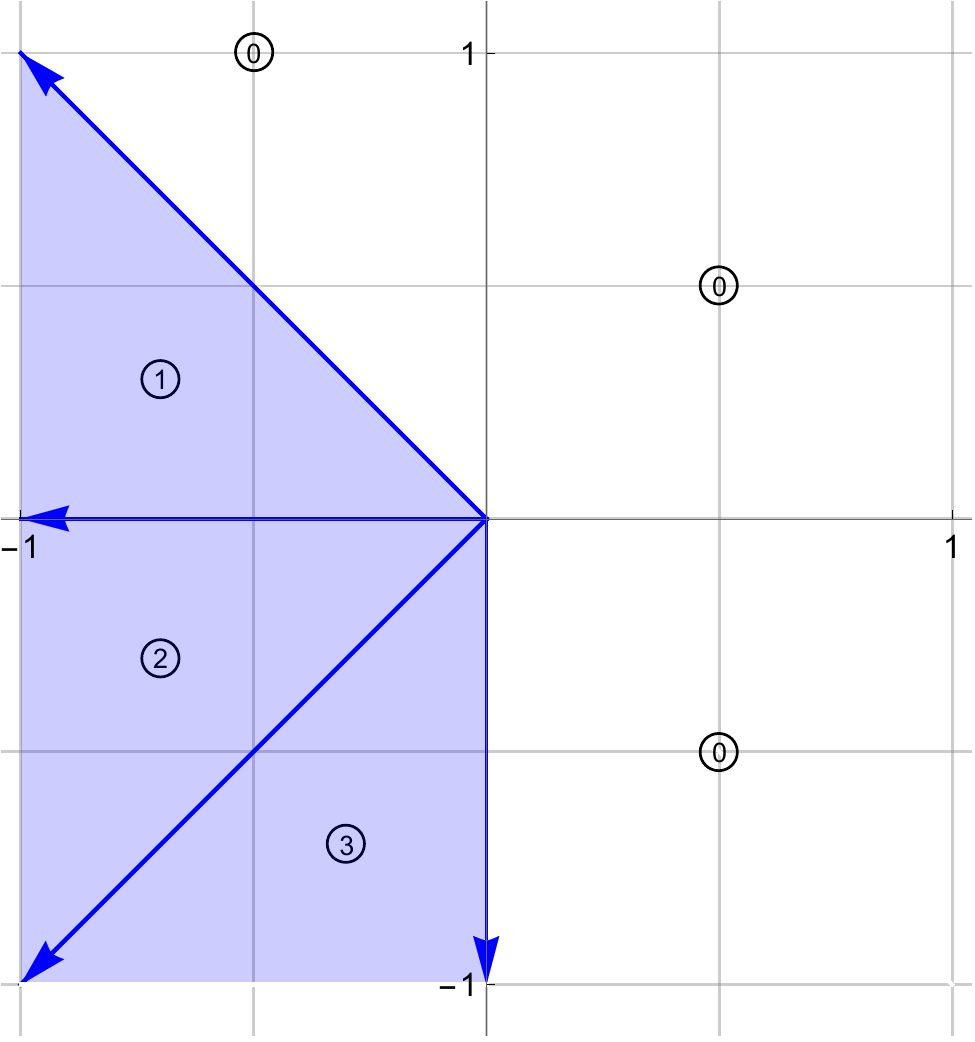}
\end{minipage}}\centering
\subfigure[Density function] {\begin{minipage}[b]{0.5\linewidth}
\includegraphics[width=1\textwidth]{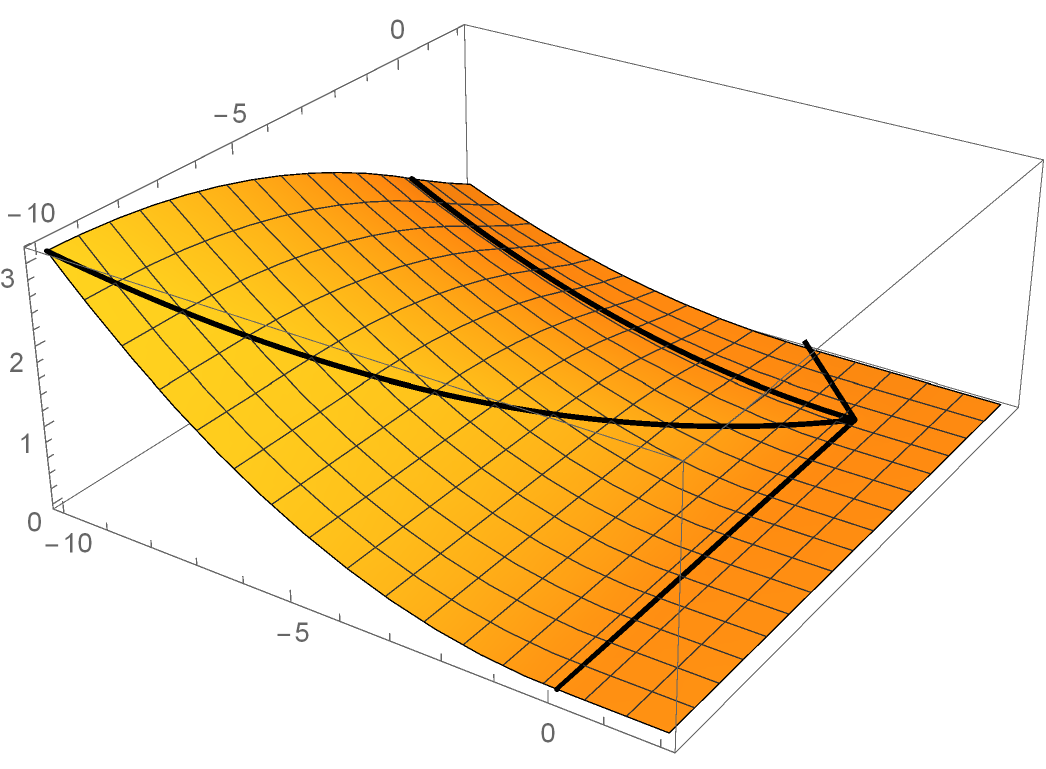}
\end{minipage}}
\caption{The support of the iterated convolution and its density
over the support}\label{fig:chamber-support}
\end{figure}

\begin{proof}
The measure $H_{(-2,2)}\star H_{(-2,0)}\star H_{(-2,-2)}\star
H_{(0,-2)}$ is, in fact, the non-Abelian Duistermaat-Heckman measure
that is on the closures of the regular chambers containing the
vertex $(0,0)$ given by the above convolution:
$$
\delta_{(0,0)}\star H_{(-2,2)}\star H_{(-2,0)}\star H_{(-2,-2)}\star
H_{(0,-2)}.
$$
Then its density is denote by $p(r,s)$. Denote
$O(0,0),P_1(-2,2),P_2(-2,0),P_3(-2,-2),P_4(0,-2)$. Then
$$
\abs{OP_1P_2}:=C_1,\quad \abs{OP_2P_3}:=C_2,\quad
\abs{OP_3P_4}:=C_3.
$$
Denote $\bbR^2-C_1-C_2-C_3:=C_0$. Thus
$$
\bbR^2=C_0\cup C_1\cup C_2 \cup C_3.
$$
(i) Clearly $p\equiv0$ on $C_0$.
\begin{itemize}
\item The wall $W_{01}$ separating $C_0$ and $C_1$ is given by the
equation: $r+s=0$. Its normal vector $\xi=(-1,-1)$. Just only one
weights $\omega_1=(-2,2)$ lies on the linear hyperplane spanned by
$W_{01}$ (other weights are outside of $W_{01}$:
$\omega_2=(-2,0),\omega_3=(-2,-2),\omega_4=(0,-2)$).
\item Consider the push-forward of Lebesgue
measure on $\bbR^1_{\geqslant0}$ along the linear map $P_{W_{01}}:
u\mapsto u\omega_1$. Its density with respect to $\dif w$ is given
by a single homogeneous polynomial on the wall $W_{01}$.
\end{itemize}
\begin{itemize}
\item Denote by $p_{W_{01}}$ any polynomial function extending it to all of
the dual Lie algebra of $\tilde T$.
\item Clearly $p_{W_{01}}=\frac12$. Indeed,
\begin{eqnarray*}
&&p^{-1}_{W_{01}}=\dif \lambda(\omega_1,\xi/\norm{\xi}^2) =
\Abs{\det\Pa{\begin{array}{cc}
               -2 & 2 \\
               -1/2 & -1/2
             \end{array}
}}=2\\
&&\Longleftrightarrow p_{W_{01}}=\frac12.
\end{eqnarray*}
\end{itemize}
Note that, during the proof, we abuse notation by using identical
symbols for a differential form and its induced measure. Hence
$\mu=(r,s)$
\begin{eqnarray*}
p(\mu)=\frac12\Res_{z=0}\Pa{\frac{e^{\Inner{\mu}{\bsx+z\xi}}}{\prod^4_{k=2}\Inner{\omega_k}{\bsx+z\xi}}}_{\bsx=\zero},
\end{eqnarray*}
where
\begin{eqnarray*}
\frac{e^{\Inner{\mu}{\bsx+z\xi}}}{\prod^4_{k=2}\Inner{\omega_k}{\bsx+z\xi}}
=
\frac{\exp\Pa{r(x_1-z)+s(x_2-z)}}{[-2(x_1-z)][[-2(x_1-z)-2(x_2-z)]][-2(x_2-z)]}.
\end{eqnarray*}
Thus
\begin{eqnarray*}
p(\mu) &=& \frac12\times\Pa{-\frac18}
\Res_{z=0}\Pa{\frac{\exp\Pa{r(x_1-z)+s(x_2-z)}}{(x_1-z)(x_1+x_2-2z)(x_2-z)}}_{\bsx=\zero}\\
&=&\frac1{32} \Res_{z=0}
\Pa{\frac{e^{-(r+s)z}}{z^3}}=\frac1{64}(r+s)^2.
\end{eqnarray*}
(ii) The wall $W_{12}$ separating $C_1$ and $C_2$ is given by the
equation: $s=0$. Its normal vector $\xi=(0,-1)$. Just only one
weights $\omega_1=(-2,0)$ lies on the linear hyperplane spanned by
$W_{12}$ (other weights are outside of $W_{12}$:
$\omega_2=(-2,2),\omega_3=(-2,-2),\omega_4=(0,-2)$).
\begin{itemize}
\item Consider the push-forward of Lebesgue measure on
$\bbR^1_{\geqslant0}$ along the linear map $P_{W_{12}}:u\mapsto
u\omega_1$. Its density with respect to $\dif w$ is given by a
single homogeneous polynomial on the wall $W_{12}$.
\item Denote by $p_{W_{12}}$ any polynomial function extending it to
all of the dual Lie algebra of $\tilde T$. Clearly
$p_{W_{12}}=\frac12$. Indeed,
\begin{eqnarray*}
p^{-1}_{W_{12}} =\dif
\lambda(\omega_1,\xi/\norm{\xi}^2)=\Abs{\det\Pa{\begin{array}{cc}
                                                                    -2 & 0 \\
                                                                    0 &
                                                                    -1
                                                                  \end{array}
}}=2\Longleftrightarrow p_{W_{12}}=\frac12.
\end{eqnarray*}
\end{itemize}
Hence
\begin{eqnarray*}
p(\mu) - \frac1{64}(r+s)^2 =
\frac12\Res_{z=0}\Pa{\frac{e^{\Inner{\mu}{\bsx+z\xi}}}{\prod^4_{k=2}\Inner{\omega_k}{\bsx+z\xi}}}_{\bsx=\zero},
\end{eqnarray*}
where
\begin{eqnarray*}
\frac{e^{\Inner{\mu}{\bsx+z\xi}}}{\prod^4_{k=2}\Inner{\omega_k}{\bsx+z\xi}}
=
\frac{e^{rx_1+s(x_2-z)}}{[-2x_1+2(x_2-z)][-2x_1-2(x_2-z)][-2(x_2-z)]}.
\end{eqnarray*}
Thus
\begin{eqnarray*}
&&p(\mu)-\frac{(r+s)^2}{64} \\
&&=
\frac12\times-\frac18\Res_{z=0}\Pa{\frac{e^{rx_1+s(x_2-z)}}{(x_1-x_2+z)(x_1+x_2-z)(x_2-z)}}_{\bsx=\zero}\\
&&= -\frac1{16}\Res_{z=0}\Pa{\frac{e^{-sz}}{z^3}} = -\frac{s^2}{32}.
\end{eqnarray*}
Therefore, on the chamber $C_2$
\begin{eqnarray*}
p(\mu) =p(r,s)= \frac{(r+s)^2}{64} - \frac{s^2}{32} =
\frac{r^2+2rs-s^2}{64}.
\end{eqnarray*}
(iii) The wall $W_{23}$ separating $C_2$ and $C_3$ is given by the
equation: $r-s=0$. Its normal vector $\xi=(1,-1)$. Just only one
weights $\omega_1=(-2,-2)$ lies on the linear hyperplane spanned by
$W_{23}$ (other weights are outside of $W_{23}$:
$\omega_2=(-2,2),\omega_3=(-2,0),\omega_4=(0,-2)$).
\begin{itemize}
\item Consider the push-forward of Lebesgue measure on
$\bbR^1_{\geqslant0}$ along the linear map $P_{W_{23}}:u\mapsto
u\omega_1$. Its density with respect to $\dif w$ is given by a
single homogeneous polynomial on the wall $W_{23}$.
\item Denote by $p_{W_{23}}$ any polynomial function extending it to
all of the dual Lie algebra of $\tilde T$. Clearly
$p_{W_{23}}=\frac12$.
\end{itemize}
Indeed,
\begin{eqnarray*}
p^{-1}_{W_{23}} = \dif
\lambda(\omega_1,\xi/\norm{\xi}^2)=\Abs{\det\Pa{\begin{array}{cc}
                                              -2 & -2 \\
                                              1/2 & -1/2
                                            \end{array}
}}=2\Longleftrightarrow p_{W_{23}}=\frac12.
\end{eqnarray*}
Hence
\begin{eqnarray*}
p(\mu) - \frac{r^2+2rs-s^2}{64} =
\frac12\Res_{z=0}\Pa{\frac{e^{\Inner{\mu}{\bsx+z\xi}}}{\prod^4_{k=2}\Inner{\omega_k}{\bsx+z\xi}}}_{\bsx=\zero},
\end{eqnarray*}
where
$$
\frac{e^{\Inner{\mu}{\bsx+z\xi}}}{\prod^4_{k=2}\Inner{\omega_k}{\bsx+z\xi}}
=
\frac{e^{r(x_1+z)+s(x_2-z)}}{[-2(x_1+z)+2(x_2-z)][-2(x_1+z)][-2(x_2-z)]}.
$$
Thus
\begin{eqnarray*}
&&p(\mu) - \frac{r^2+2rs-s^2}{64} \\
&&=\frac12\times
-\frac18\Res_{z=0}\Pa{\frac{e^{r(x_1+z)+s(x_2-z)}}{(x_1-x_2+2z)(x_1+z)(x_2-z)}}_{\bsx=\zero}\\
&&=\frac1{32}\Res_{z=0}\Pa{\frac{e^{(r-s)z}}{z^3}} =
\frac{(r-s)^2}{64}.
\end{eqnarray*}
Therefore, on the chamber $C_3$
\begin{eqnarray*}
p(\mu)=p(r,s)=\frac{r^2+2rs-s^2}{64}+\frac{(r-s)^2}{64} =
\frac{r^2}{32}.
\end{eqnarray*}
This completes the proof.
\end{proof}

\section{Applications of Duistermaat-Heckman measure}\label{sect:6}

For $\lambda=(\lambda_1,\lambda_2,\lambda_3,\lambda_4)\in\bbR^3_{\geqslant0}$ with $\lambda_1>\lambda_2>\lambda_3>\lambda_4$ and $\sum^4_{k=1}\lambda_k=1$, denote the $\SU(4)$-adjoint orbit of $\Lambda=\diag(\lambda_1,\lambda_2,\lambda_3,\lambda_4)$ by
$$
\cU_{\Lambda}:=\Set{\bsU\Lambda\bsU^\dagger:\bsU\in\SU(4)}.
$$
 Let
$$
\hat\Lambda: = \Lambda-\frac{\I_4}
4=\diag(\hat\lambda_1,\hat\lambda_2,\hat\lambda_3,\hat\lambda_4).
$$
where $\hat\lambda_k=\lambda_k-\frac14$. Denote by
$\hat\lambda=(\hat\lambda_1,\hat\lambda_2,\hat\lambda_3,\hat\lambda_4)$.
Apparently $\hat\lambda_1>\hat\lambda_2>\hat\lambda_3>\hat\lambda_4$
and $\sum^4_{k=0}\hat\lambda_k=0$. It is easily seen that
$$
\hat\Lambda\in \mathrm{i}\h,
$$
where $\liet$ is the Cartan subalgebra of $\su(4)$, the Lie algebra
of $\SU(4)$. By chosing $\Ad$-invariant inner product on $\liet$, we
can identify $\liet^*$ with $\liet$. Let $\tilde\liet$ be the Lie
algebra of maximal torus $\tilde T$ of $\tilde K=\SU(2)\times
\SU(2)$.

Denote
\begin{eqnarray}\label{eq:ccoef}
c_3=2(\hat\lambda_1+\hat\lambda_2),\quad
c_2=2(\hat\lambda_1+\hat\lambda_3),\quad
c_1=2\Abs{\hat\lambda_1+\hat\lambda_4}.
\end{eqnarray}
From the above definition, we infer that
\begin{eqnarray}
c_3>c_2>c_1>0.
\end{eqnarray}
\begin{prop}
Consider the Lie group $\tilde K = \SU(2)\times\SU(2)$ with its
maximal torus $\tilde T$ whose Lie algebra is given by
$\tilde\liet$. The positive Weyl chamber
$\mathrm{i}\tilde\liet^*_{\geqslant0}$ can be identified with
$\bbR^2_{\geqslant0}$.
\end{prop}

\begin{proof}
The Lie group $\SU(2)$ has a rank one Lie algebra, $\su(2)$, with
Cartan subalgebra $\h$ spanned by $\bsH=\proj{0}-\proj{1}$. The root
system is of a single positive root $\alpha$ satisfying
$\alpha(\bsH)=2$. The positive Weyl chamber is a subset of the
Cartan subalgebra $\h$ defined by the condition that all positive
roots take nonnegative values. Any element $h\in\mathrm{i}\h$ can be
expressed as $h=z\bsH$ for some $z\in\bbR$. Then
$\alpha(h)=z\alpha(\bsH)=2z\geqslant0$. Geometrically,
$\h\cong\bbR$. Thus, the positive Weyl chamber for $\SU(2)$ is the
set of all nonnegative multipliers of $\bsH$, i.e.,
\begin{eqnarray*}
\h_{\geqslant0}=\Set{z\bsH:z\geqslant0}\cong\bbR_{\geqslant0}.
\end{eqnarray*}
The Lie group $\tilde K=\SU(2)\times\SU(2)$ has rank two. Its Lie
algebra is $\tilde\k=\su(2)\oplus \su(2)$, and the Cartan subalgebra
$\tilde\liet$ is spanned by the generators $\bsH_1$ from the first
$\su(2)$ factor and $\bsH_2$ from the second $\su(2)$ factor. Using
the standard representation, the generators are
$$
\bsH_1=\bsH\ot\proj{0}\text{ and }\bsH_2=\bsH\ot\proj{1}.
$$
The root system of $\tilde K=\SU(2)\times\SU(2)$ is the disjoint
union of the root systems of each $\SU(2)$ factor. The positive
roots are $\alpha_1$ for the first factor and $\alpha_2$ for the
second factor, defined
$$
\alpha_1(\bsH_1)=2,\alpha_1(\bsH_2)=0,
\alpha_2(\bsH_1)=0,\alpha_2(\bsH_2)=2.
$$
The positive Weyl chamber is the subset of the Cartan subalgebra
where all positive roots take non-negative values. For an element
$h=z_1\bsH_1+z_2\bsH_2\in\mathrm{i}\tilde\liet$, the conditions are
\begin{eqnarray*}
\alpha_1(h) =\alpha_1(z_1\bsH_2+z_2\bsH_2)=
2z_1+0z_2=2z_1\geqslant0,\\
\alpha_2(h) =\alpha_2(z_1\bsH_2+z_2\bsH_2)=
0z_1+2z_2=2z_2\geqslant0.
\end{eqnarray*}
Thus, the positive Weyl chamber is identified with
\begin{eqnarray*}
\mathrm{i}\tilde\liet_{\geqslant0}=\Set{z_1\bsH_1+z_2\bsH_2:z_1\geqslant0,z_2\geqslant0}\cong\bbR^2_{\geqslant0}.
\end{eqnarray*}
We are done.
\end{proof}

With $\mathrm{i}\tilde\liet^*_{\geqslant0}\cong\bbR^2_{\geqslant0}$,
we have already seen that
\begin{eqnarray}
\dh^{\tilde
K}_{\cO_{\hat\lambda}}\Big|_{\mathrm{i}\tilde\liet^*_{\geqslant0}} =
\sum_{w\in S_4}\sign(w)\delta_{\pi(w\hat\lambda)}\star
H_{(-2,0)}\star H_{(-2,-2)}\star H_{(0,-2)}\star
H_{(-2,2)}\Big|_{\mathrm{i}\tilde\liet^*_{\geqslant0}}
\end{eqnarray}
is supported on $\mathrm{i}\tilde\liet^*_{\geqslant0}$, which
amounts to be also supported on $\bbR^2_{\geqslant0}$. We can write
it in the following form:
\begin{prop}
For $\tilde K=\SU(2)\times\SU(2)$, it holds that
\begin{eqnarray}
\dh^{\tilde K}_{\cO_{\hat\lambda}}\Big|_{\bbR^2_{\geqslant0}} =
\sum_{w\in S_4}\sign(w)\delta_{\pi(w\hat\lambda)}\star
H_{(-2,0)}\star H_{(-2,-2)}\star H_{(0,-2)}\star
H_{(-2,2)}\Big|_{\bbR^2_{\geqslant0}}
\end{eqnarray}
with its density function being given by
\begin{eqnarray}
&&\sum_{w\in S_4}\sign(w)\delta_{\pi(w\hat\lambda)}\star
p(x,y)\Big|_{\bbR^2_{\geqslant0}}\notag\\
&&=\Pa{\delta_{(c_2,\pm c_1)}-\delta_{(c_1,c_2)} +
\delta_{(c_1,c_3)}-\delta_{(c_3,\pm c_1)} + \delta_{(c_3,\pm
c_2)}-\delta_{(c_2,c_3)}}\star
p(x,y)\Big|_{\bbR^2_{\geqslant0}},\label{eq:momentpoly}
\end{eqnarray}
where the expression of $p(x,y)$ is from
Proposition~\ref{prop:4density}.
\end{prop}

\begin{proof}
Sketch of proof can be given for Eq.~\eqref{eq:momentpoly}: Indeed,
we have already known the support of $p(x,y)$ is $\cup^3_{k=1}C_k$.
Then the support of $(\delta_{\pi(w\hat\lambda)}\star p)(x,y)$ is
the shifted version of $\cup^3_{k=1}C_k$ to the point
$\pi(w\hat\lambda)$. However, points whose corresponding shifted
supports do not intersect with $\bbR^2_{>0}$ do not contribute to
the density. In other words, by calculation, we find that it remains
only 9 points can contribute to the density.
\end{proof}

\begin{remark}
Moreover, the non-Abelian moment polytope for the action of $\tilde
K=\SU(2)\times\SU(2)$ on a generic co-adjoint $K=\SU(4)$-orbit can
be identified as
\begin{eqnarray}\label{eq:momentpolytope}
\Delta_{\tilde
K}(\cO_{\hat\lambda}):=\Set{(x,y)\in[0,c_3]^2:x+y\leqslant
c_2+c_3,\abs{x-y}\leqslant c_3-c_1},
\end{eqnarray}
where $c_k$'s are from Eq.~\eqref{eq:ccoef}. This is just the
support of the density in Eq.~\eqref{eq:momentpoly}. Based on the
constraint Eq.~\eqref{eq:momentpolytope}, we can re-derive the
compatibility solution for two-qubit quantum marginal problem. Let
me describe it below: Let $\rho_{AB}\in\rD(\bbC^2\ot\bbC^2)$ with
two reduced states $(\rho_A,\rho_B)$, where the eigenvalues of
$\rho_{AB}$ are ordered non-increasingly as
$\lambda_1\geqslant\cdots\geqslant\lambda_4\geqslant0$, and the
minimal eigenvalue of $\rho_X(X=A,B)$ is denoted by
$\lambda_{\min}(\rho_X)$. Via $\hat\lambda_j = \lambda_j-\frac14$
for $j\in\set{1,2,3,4}$ and $x=1-2\lambda_{\min}(\rho_1)$ and
$y=1-2\lambda_{\min}(\rho_2)$, we see that the condition
$x\in[0,c_3]$ can be rewritten as
\begin{eqnarray}\label{eq:compatibility1}
\min\Set{\lambda_{\min}(\rho_A),\lambda_{\min}(\rho_B)}\geqslant
\lambda_3(\rho_{AB})+\lambda_4(\rho_{AB});
\end{eqnarray}
and the condition $x+y\leqslant c_2+c_3$ can be rewritten as
\begin{eqnarray}\label{eq:compatibility2}
\lambda_{\min}(\rho_A)+\lambda_{\min}(\rho_B)\geqslant
\lambda_2(\rho_{AB})+\lambda_3(\rho_{AB})+2\lambda_4(\rho_{AB});
\end{eqnarray}
and the condition $\abs{x-y}\leqslant c_3-c_1$ can be rewritten as
\begin{eqnarray}\label{eq:compatibility3}
\Abs{\lambda_{\min}(\rho_A)-\lambda_{\min}(\rho_B)}\geqslant
\min\set{\lambda_1(\rho_{AB})-\lambda_3(\rho_{AB}),
\lambda_2(\rho_{AB})-\lambda_4(\rho_{AB})}.
\end{eqnarray}
In summary, for any given 2-tuple $(\rho_A,\rho_B)$ of qubit states,
there exists a global two-qubit states $\rho_{AB}$ such that
$\rho_A=\ptr{A}{\rho_{AB}}$ and $\rho_B=\ptr{A}{\rho_{AB}}$ if and
only if the three compatibility conditions
Eqs.~\eqref{eq:compatibility1}--\eqref{eq:compatibility3} hold. This
result is obtained already by Bravyi \cite{Bravyi2004}.
\end{remark}

Consider the pushforward measure $F_*(\mu_{\cO_{\hat\lambda}})$ of
Liouville measure $\mu_{\cO_{\hat\lambda}}$ along the map $F:
\cO_{\hat\lambda}\to\blue{\mathrm{i}\h^*_{\geqslant0}}\cong\bbR_{\geqslant0}$,
where $\h^*_{\geqslant0}$ is the dual Cartan subalgebra of $\su(2)$
and $F=q_1\circ\pi\circ\Phi$,
\begin{eqnarray}\label{eq:pushforwardbyF}
F:\cO_{\hat\lambda}(\subset
\k^*)\stackrel{\Phi}{\hookrightarrow}\mathrm{i}\k^*\stackrel{\pi}{\to}\mathrm{i}\tilde\liet^*_{\geqslant0}=\blue{\mathrm{i}\h^*_{\geqslant0}}\oplus
\mathrm{i}\h^*_{\geqslant0}\stackrel{q_1}{\to}\blue{\mathrm{i}\h^*_{\geqslant0}}.
\end{eqnarray}
\begin{prop}
The pushforward measure $F_*\mu_{\cO_{\hat\lambda}}$ is just
\begin{eqnarray}
F_*(\mu_{\cO_{\hat\lambda}}) =
(q_1)_*\pi_*\Phi_*(\mu_{\cO_{\hat\lambda}}) = (q_1)_*\Pa{p_{\tilde
K}\dh^{\tilde K}_{\cO_{\hat\lambda}}}
\end{eqnarray}
whose density function is given by $(2\pi)^6\cI(x|\hat\lambda)$,
where
\begin{eqnarray*}
\cI(x|\hat\lambda):=\int_{\bbR^2_{\geqslant0}}xy\Br{\Pa{\delta_{(c_2,\pm
c_1)}-\delta_{(c_1,c_2)} + \delta_{(c_1,c_3)}-\delta_{(c_3,\pm c_1)}
+ \delta_{(c_3,\pm c_2)}-\delta_{(c_2,c_3)}}\star p}(x,y)\dif x\dif
y,
\end{eqnarray*}
where the expression of $p(x,y)$ is from
Proposition~\ref{prop:4density}. Moreover, it holds that
\begin{eqnarray}\label{eq:RNderivative}
\cI(x|\hat\lambda)=\cI_1(x|\hat\lambda)+\cI_2(x|\hat\lambda)+\cI_3(x|\hat\lambda),
\end{eqnarray}
where
\begin{eqnarray}
\cI_1(x|\hat\lambda) &=& \frac{c^2_3-c^2_2}{64}x(x-c_1)^2\chi_{[0,c_1]}(x),\\
\cI_2(x|\hat\lambda) &=& \frac{c^2_1-c^2_3}{64}x(x-c_2)^2\chi_{[0,c_2]}(x),\\
\cI_3(x|\hat\lambda) &=&
\frac{c^2_2-c^2_1}{64}x(x-c_3)^2\chi_{[0,c_3]}(x),
\end{eqnarray}
where $c_k$'s are from Eq.~\eqref{eq:ccoef}, and $\chi_A(x)$ is the
indicator of a set $A$.
\end{prop}

\begin{proof}
Note that $p_{\tilde K}(\frac x2\bsh,\frac y2\bsh)=p_{\SU(2)}(\frac
x2\bsh)p_{\SU(2)}(\frac y2\bsh)=xy$. We consider the following
integral
\begin{eqnarray*}
&&\cI(x|\hat\lambda):=\int_{\bbR^2_{\geqslant0}}xy\Br{\Pa{\delta_{(c_2,\pm
c_1)}-\delta_{(c_1,c_2)} + \delta_{(c_1,c_3)}-\delta_{(c_3,\pm c_1)}
+ \delta_{(c_3,\pm
c_2)}-\delta_{(c_2,c_3)}}\star p}(x,y)\dif x\dif y\\
&&=\int_{\bbR^2_{\geqslant0}}xy\Br{\Pa{\delta_{(c_1,c_3)}-\delta_{(c_1,c_2)}
}\star p}(x,y)\dif x\dif y \\
&&~~~+
\int_{\bbR^2_{\geqslant0}}xy\Br{\Pa{-\delta_{(c_2,c_3)}+\delta_{(c_2,c_1)}+\delta_{(c_2, -c_1)}}\star p}(x,y)\dif x\dif y\\
&&~~~+\int_{\bbR^2_{\geqslant0}}xy\Br{\Pa{\delta_{(c_3,
c_2)}+\delta_{(c_3,- c_2)}-\delta_{(c_3,
c_1)}-\delta_{(c_3,-c_1)}}\star p}(x,y)\dif x\dif y
\end{eqnarray*}
which is the density of the Duistermaat-Heckman measure with respect
to the Lebesgue measure on the positive Weyl chamber (which is
identified with $\bbR^2_{\geqslant0}$) of $\SU(2)\times\SU(2)$.
Denote by
\begin{eqnarray}
\cI_1(x|\hat\lambda) &=& \int_{\bbR^2_{\geqslant0}}xy\Br{\Pa{\delta_{(c_1,c_3)}-\delta_{(c_1,c_2)}}\star p}(x,y)\dif x\dif y,\\
\cI_2(x|\hat\lambda) &=&\int_{\bbR^2_{\geqslant0}}xy\Br{\Pa{-\delta_{(c_2,c_3)}+\delta_{(c_2,c_1)}+\delta_{(c_2, -c_1)}}\star p}(x,y)\dif x\dif y,\\
\cI_3(x|\hat\lambda)
&=&\int_{\bbR^2_{\geqslant0}}xy\Br{\Pa{\delta_{(c_3,
c_2)}+\delta_{(c_3,- c_2)}-\delta_{(c_3,
c_1)}-\delta_{(c_3,-c_1)}}\star p}(x,y)\dif x\dif y.
\end{eqnarray}
Note that
$c_1=2\abs{\hat\lambda_1+\hat\lambda_4}=2\abs{\hat\lambda_2+\hat\lambda_3}$
because
$\hat\lambda_1+\hat\lambda_4=-(\hat\lambda_2+\hat\lambda_3)$. Thus
$$
c_1=\max\set{\hat\lambda_1+\hat\lambda_4,\hat\lambda_2+\hat\lambda_3}.
$$
\begin{enumerate}[(a)]
\item The range $[0,c_3]$ of parameter $x$ can be decomposed into
two parts $[0,c_3]=[0,c_1]\cup[c_1,c_3]$. In order to calculate
$\cI_1$, note that $\cI_1\equiv0$ if $x\in[c_1,c_3]$, it suffices to
calculate it for $x\in[0,c_1]$. Thus
\begin{eqnarray*}
\cI_1(x|\hat\lambda)& =&\int^{x+c_3-c_1}_{0}xy \cdot
p_3(x-c_1,y-c_3)\dif y
        +\int^{c_3}_{x+c_3-c_1}xy\cdot p_2 (x-c_1,y-c_3)\dif y\\
        &&+\int^{-x+c_3+c_1}_{c_3}xy\cdot p_1(x-c_1,y-c_3)\dif y-\int^{x+c_2-c_1}_{0}xy\cdot p_3(x-c_1,y-c_2)\dif y\\
        &&-\int^{c_2}_{x+c_2-c_1}xy\cdot p_2(x-c_1,y-c_2)\dif y-\int^{-x+c_1+c_2}_{c_2}xy\cdot p_1(x-c_1,y-c_2)\dif y\\
        &=& \frac{c^2_3-c^2_2}{64}x(x-c_1)^2,
\end{eqnarray*}
which leads to
\begin{eqnarray}
\cI_1(x|\hat\lambda)=
\begin{cases}
\frac{c^2_3-c^2_2}{64}x(x-c_1)^2,&\text{if }x\in[0,c_1],\\
0, &\text{if }x\in [c_1,c_3].
\end{cases}
\end{eqnarray}
With the help of the following Figure~\ref{fig:2}, the above
calculation is easily obtained.
\begin{figure}[ht]\centering
{\begin{minipage}[b]{1\linewidth}
\includegraphics[width=1\textwidth]{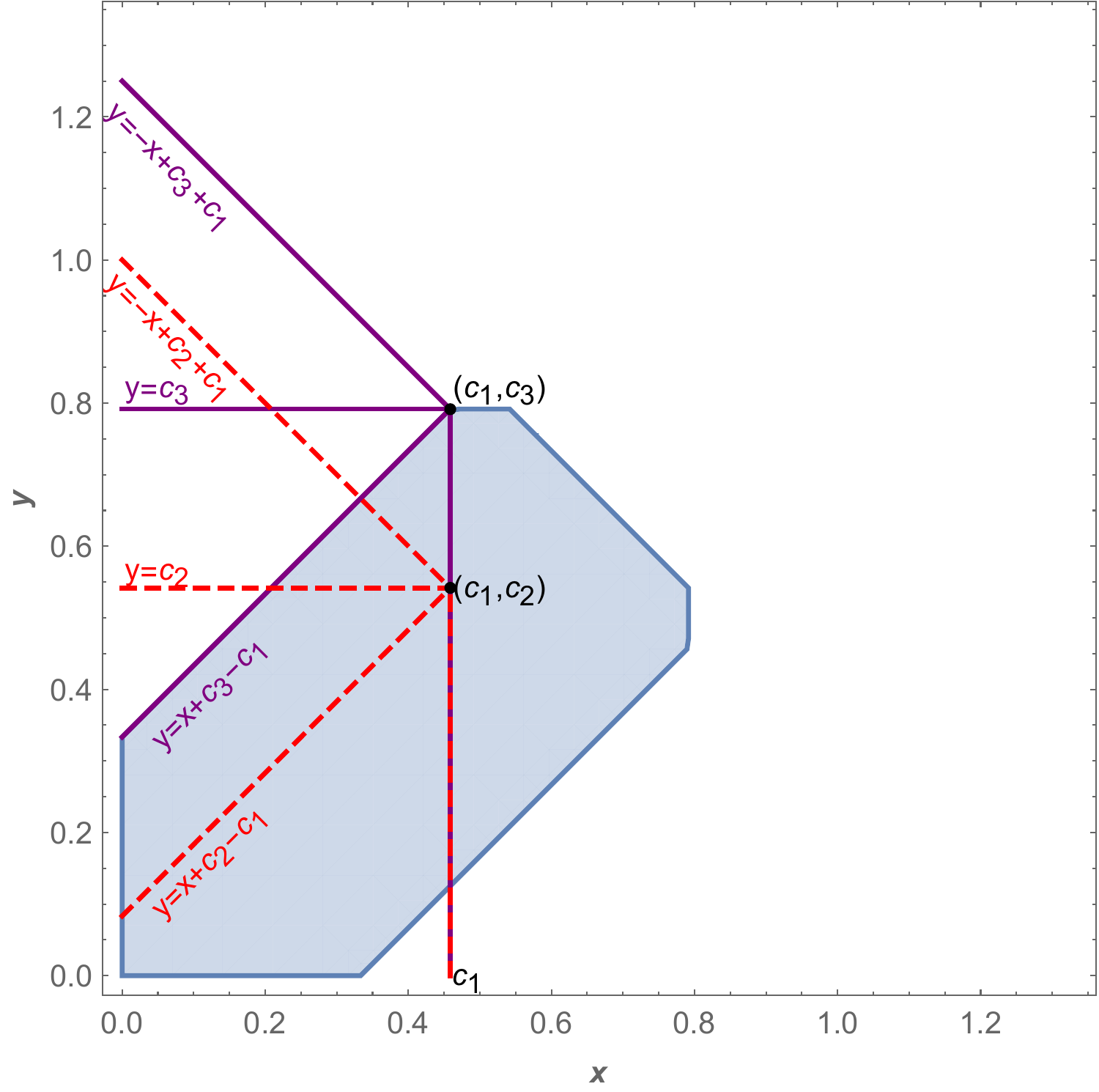}
\end{minipage}}
\caption{The calculation of $\cI_1$ over the non-Abelian moment
polytope.}\label{fig:2}
\end{figure}
\item The range $[0,c_3]$ of parameter $x$ can be decomposed into three
parts: $[0,c_3]=[0,c_2-c_1]\cup[c_2-c_1,c_2]\cup[c_2,c_3]$. Note
that $\cI_2\equiv0$ when $x\in[c_2,c_3]$. It suffices to calculate
$\cI_2$ on $[0,c_2-c_1]\cup[c_2-c_1,c_2]$. See the following
Figure~\ref{fig:3}.
\begin{enumerate}
\item[(b1)] For
$x\in[0,c_2-c_1]$, we see that
\begin{eqnarray*}
\cI_2(x|\hat\lambda) &=& -\int^{x+c_3-c_2}_0 xy p_3(x-c_2,y-c_3)\dif
y -
\int^{c_3}_{x+c_3-c_2}xy p_2(x-c_2,y-c_3)\dif y\\
&&-\int^{-x+c_3+c_2}_{c_3}xy p_1(x-c_2,y-c_3)\dif y + \int^{c_1}_0xy
p_2(x-c_2,y-c_1)\dif y\\
&&+\int^{-x+c_2+c_1}_{c_1}xy p_3(x-c_2,y-c_1)\dif y +
\int^{-x+c_2-c_1}_0xy p_3(x-c_2,y+c_1)\dif y\\
&=& \frac{c^2_1-c^2_3}{64}x(x-c_2)^2.
\end{eqnarray*}
\item[(b2)] For $x\in[c_2-c_1,c_2]$,
\begin{eqnarray*}
\cI_2(x|\hat\lambda) &=& -\int^{x+c_3-c_2}_0 xy p_3(x-c_2,y-c_3)\dif
y -
\int^{c_3}_{x+c_3-c_2}xy p_2(x-c_2,y-c_3)\dif y\\
&&-\int^{-x+c_3+c_2}_{c_3}xy p_1(x-c_2,y-c_3)\dif y +
\int^{x+c_1-c_2}_0xy
p_3(x-c_2,y-c_1)\dif y\\
&&+\int^{c_1}_{x+c_1-c_2}xy p_2(x-c_2,y-c_1)\dif y +
\int^{-x+c_2+-c_1}_{c_1}xy p_1(x-c_2,y-c_1)\dif y\\
&=& \frac{c^2_1-c^2_3}{64}x(x-c_2)^2.
\end{eqnarray*}
\end{enumerate}
This leads to
\begin{eqnarray}
\cI_2(x|\hat\lambda)=
\begin{cases}
\frac{c^2_1-c^2_3}{64}x(x-c_2)^2,&\text{if }x\in[0,c_2],\\
0, &\text{if }x\in [c_2,c_3].
\end{cases}
\end{eqnarray}
\begin{figure}[ht]\centering
{\begin{minipage}[b]{1\linewidth}
\includegraphics[width=1\textwidth]{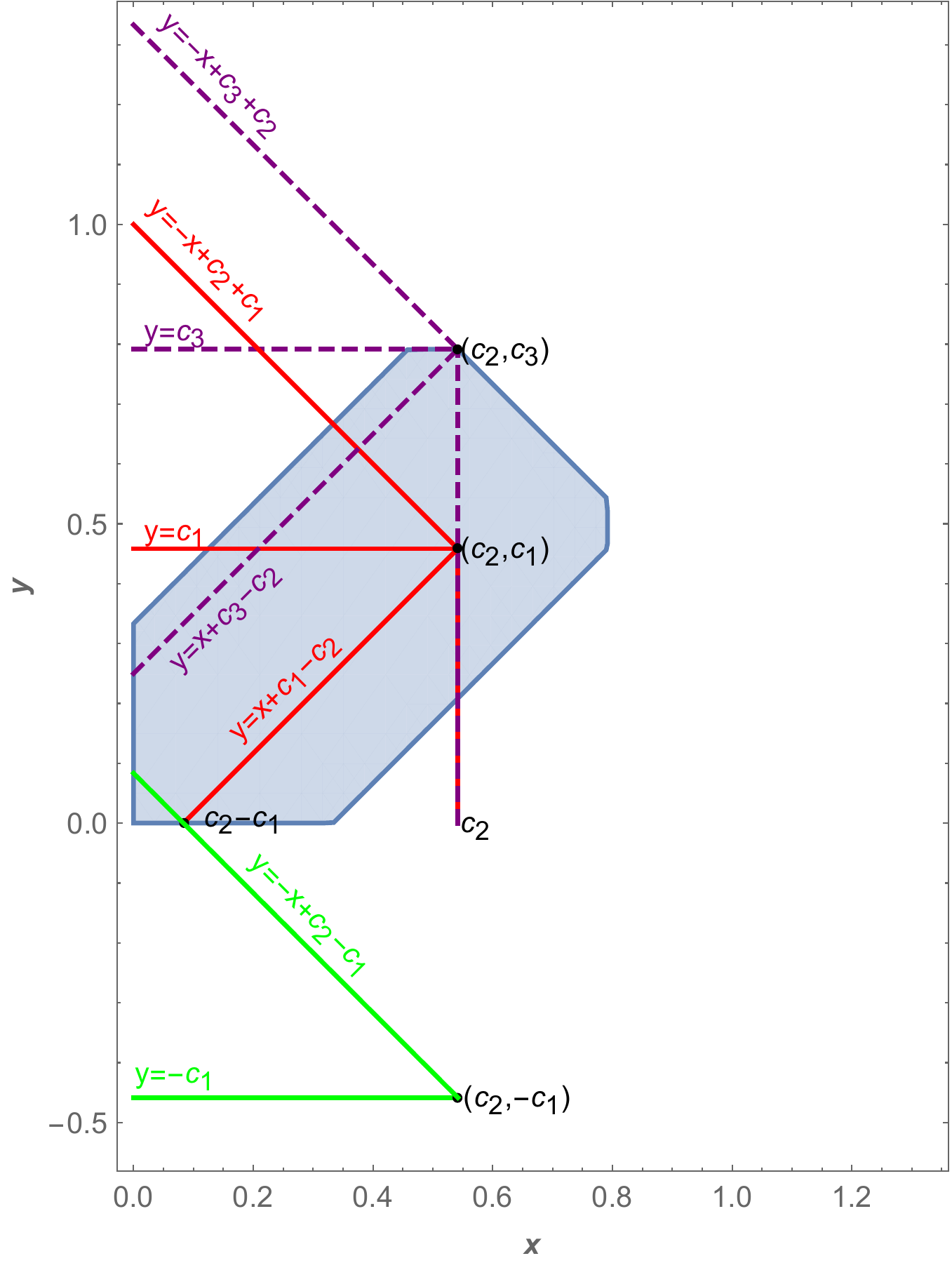}
\end{minipage}}
\caption{The calculation of $\cI_2$ over the non-Abelian moment
polytope.}\label{fig:3}
\end{figure}
\item The range $[0,c_3]$ of parameter $x$ can be decomposed into three
parts: $[0,c_3]=[0,c_3-c_2]\cup[c_3-c_2,c_3-c_1]\cup[c_3-c_1,c_3]$.
It suffices to calculate $\cI_2$ on $[0,c_2-c_1]\cup[c_2-c_1,c_2]$.
See the following Figure~\ref{fig:4}.
\begin{enumerate}
\item[(c1)] For $x\in[0,c_3-c_2]$, we see that
\begin{eqnarray*}
\cI_3(x|\hat\lambda) &=& \int^{c_2}_0xy p_2(x-c_3,y-c_2)\dif x\dif y
+ \int^{-x+c_3+c_2}_{c_2}xy
p_1(x-c_3,y-c_2)\dif x\dif y\\
&&+ \int^{-x+c_3-c_2}_0xy p_1(x-c_3,y+c_2)\dif x\dif y -
\int^{c_1}_0xy p_2(x-c_3,y-c_1)\dif x\dif y\\
&&-\int^{-x+c_3+c_1}_{c_1}xy p_1(x-c_3,y-c_1)\dif x\dif
y-\int^{-x+c_3-c_1}_0xy p_1(x-c_3,y+c_1)\dif x\dif y\\
&=& \frac{c^2_2-c^2_1}{64}x(x-c_3)^2.
\end{eqnarray*}
\item[(c2)] For $x\in[c_3-c_2,c_3-c_1]$,
\begin{eqnarray*}
\cI_3(x|\hat\lambda) &=& \int^{x+c_2-c_3}_0 xy p_3(x-c_3,y-c_2)\dif
x\dif y +
\int^{c_2}_{x+c_2-c_3}xy p_2(x-c_3,y-c_2)\dif x\dif y\\
&&+\int^{-x+c_3+c_2}_{c_2}xy p_1(x-c_3,y-c_2)\dif x\dif y -
\int^{c_1}_0xyp_2(x-c_3,y-c_1)\dif x\dif y\\
&&- \int^{-x+c_3+c_1}_{c_1}xyp_1(x-c_3,y-c_1)\dif x\dif y
-\int^{-x+c_3-c_1}_0xy p_1(x-c_3,y+c_1)\dif x\dif y\\
&=& \frac{c^2_2-c^2_1}{64}x(x-c_3)^2.
\end{eqnarray*}
\item[(c3)] For $x\in[c_3-c_1,c_3]$,
\begin{eqnarray*}
\cI_3(x|\hat\lambda) &=& \int^{x+c_2-c_3}_0xy p_3(x-c_3,y-c_2)\dif
x\dif y + \int^{c_2}_{x+c_2-c_3}xy p_2(x-c_3,y-c_2)\dif
x\dif y \\
&&+ \int^{-x+c_2+c_3}_{c_2}xy p_1(x-c_3,y-c_2)\dif x\dif y
-\int^{x+c_1-c_3}_0xy p_3(x-c_3,y-c_1)\dif x\dif y\\
&&-\int^{c_1}_{x+c_1-c_3}xy p_2(x-c_3,y-c_1)\dif x\dif
y-\int^{-x+c_3+c_1}_{c_1}xy p_1(x-c_3,y-c_1)\dif x\dif y\\
&&=\frac{c^2_2-c^2_1}{64}x(x-c_3)^2.
\end{eqnarray*}
\end{enumerate}
These leads to
\begin{eqnarray}
\cI_3(x|\hat\lambda)=\frac{c^2_2-c^2_1}{64}x(x-c_3)^2,\quad
(x\in[0,c_3]).
\end{eqnarray}
\begin{figure}[ht]\centering
{\begin{minipage}[b]{1\linewidth}
\includegraphics[width=1\textwidth]{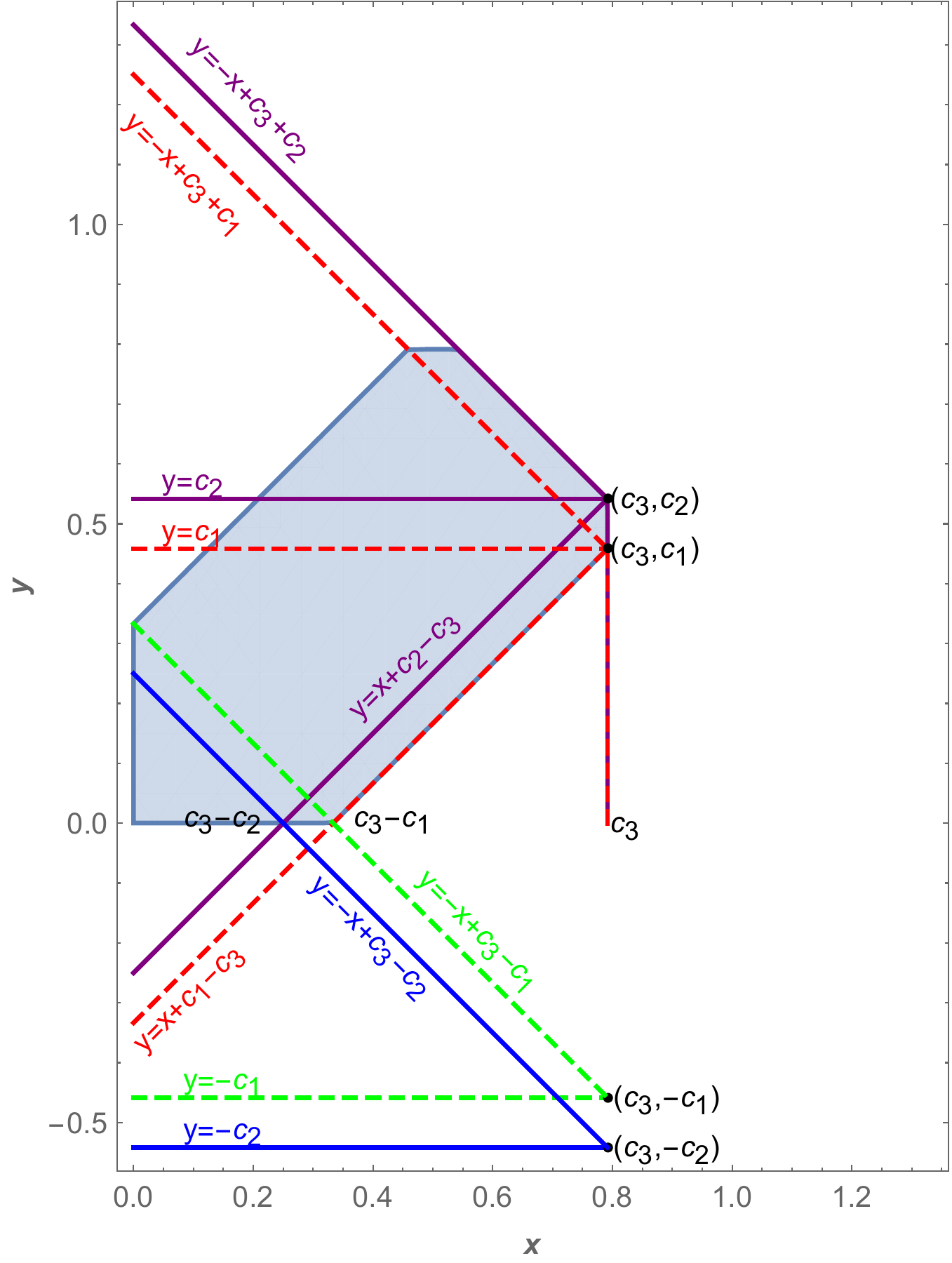}
\end{minipage}}
\caption{The calculation of $\cI_3$ over the non-Abelian moment
polytope.}\label{fig:4}
\end{figure}
\end{enumerate}
This completes the proof.
\end{proof}

\subsection{Separability probability in the conditioned state space}

For a given qudit state $\eta\in\density{\bbC^n}$, Milz and Strunz
considered the conditioned state space \cite{Milz2015}
\begin{eqnarray}
\rD^\eta(\bbC^n\ot\bbC^m) := \Set{\rho\in \density{\bbC^n\ot\bbC^m}:
\Ptr{\bbC^m}{\rho}=\eta\in \density{\bbC^n}}
\end{eqnarray}
and the set of separable states in $\rD^\eta(\bbC^n\ot\bbC^m)$,
denoted by $\rD^\eta_\sep(\bbC^n\ot\bbC^m)$. They studied the
separability probability in the conditioned state space
\begin{eqnarray}\label{eq:sepprobconditionedstatespace}
P^{(n\times m)}_{\sep}(\eta) :=
\frac{\vol_{\rH\rS}\Pa{\rD^\eta_\sep(\bbC^n\ot\bbC^m)}}{\vol_{\rH\rS}\Pa{\rD^\eta(\bbC^n\ot\bbC^m)}}.
\end{eqnarray}
All further considerations will be simplified by the observation
that both the Hilbert-Schmidt measure in the set
$\rD^{\eta}(\bbC^n\ot\bbC^m)$ and the separability of a state
$\rho\in\rD^\eta(\bbC^n\ot\bbC^m)$ are invariant under a
transformation $\bsU\ot\I_m$, where $\bsU\in\SU(n)$. Indeed,
\begin{eqnarray}
\rD^{\bsU\eta\bsU^\dagger}(\bbC^n\ot\bbC^m)
=(\bsU\ot\I_m)\rD^{\eta}(\bbC^n\ot\bbC^m)(\bsU\ot\I_m)^\dagger.
\end{eqnarray}
For $n=2$, the qubit state $\eta\in\rD(\bbC^2)$ can be represented
by Bloch vector $\eta=\frac12(\I_2+\bsa\cdot\boldsymbol{\sigma})$.
We denote by $\eta\approx \bsa$ simply. Thus we make the following
identifications:
\begin{eqnarray*}
\rD^{\bsa}(\bbC^2\ot\bbC^m) &=& \rD^{\eta}(\bbC^2\ot\bbC^m),\\
\rD^{\bsa}_\sep(\bbC^2\ot\bbC^m) &=&
\rD^{\eta}_\sep(\bbC^2\ot\bbC^m),\\
P^{(2\times m)}_{\sep}(\bsa)&=& P^{(2\times m)}_{\sep}(\eta).
\end{eqnarray*}
Thus Eq.~\eqref{eq:sepprobconditionedstatespace}, can be represented
as
\begin{eqnarray}
P^{(2\times m)}_{\sep}(\bsa) :=
\frac{\vol_{\rH\rS}\Pa{\rD^{\bsa}_\sep(\bbC^2\ot\bbC^m)}}{\vol_{\rH\rS}\Pa{\rD^{\bsa}(\bbC^2\ot\bbC^m)}}.
\end{eqnarray}
We have already known the fact that $\SU(2)$ is the double cover of
$\SO(3)$ from Lie Theory, which leads to the following
$\bsU\eta\bsU^\dagger=\frac12(\I_2+(\bsO\bsa)\cdot\boldsymbol{\sigma})$,
where both $\bsU\in\SU(2)$ and $\bsO=(o_{ij})\in\SO(3)$ are
connected via $o_{ij}=\frac12\Tr{\sigma_i\bsU\sigma_j\bsU^\dagger}$.
\begin{prop}\label{prop:blochlengtha}
It holds that
\begin{eqnarray}
P^{(2\times m)}_{\sep}(\bsO\bsa) =P^{(2\times m)}_{\sep}(\bsa)
\end{eqnarray}
for any $\bsO\in\SO(3)$ and $\bsa\in\bbR^3$ with
$\abs{\bsa}\leqslant1$.
\end{prop}

\begin{proof}
In fact, it is easily seen that
\begin{eqnarray*}
\vol_{\rH\rS}\Pa{(\bsU\ot\I_m)\rD^{\eta}(\bbC^n\ot\bbC^m)(\bsU\ot\I_m)^\dagger}
&=& \vol_{\rH\rS}\Pa{\rD^{\eta}(\bbC^n\ot\bbC^m)}\\
\vol_{\rH\rS}\Pa{(\bsU\ot\I_m)\rD^{\eta}_\sep(\bbC^n\ot\bbC^m)(\bsU\ot\I_m)^\dagger}
&=& \vol_{\rH\rS}\Pa{\rD^{\eta}_\sep(\bbC^n\ot\bbC^m)}
\end{eqnarray*}
holds for all $\bsU\in\SU(2)$, which means that $P^{(2\times
m)}_{\sep}(\bsU\eta\bsU^\dagger)=P^{(2\times m)}_{\sep}(\eta)$. The
desired result is obtained by translating it into the Bloch vector
form.
\end{proof}

Let $a=\abs{\bsa}\in[0,1]$. From the above
Proposition~\ref{prop:blochlengtha}, we see that $P^{(2\times
m)}_{\sep}(\bsa)$ is constant on the sphere
$S_a:=\Set{\bsa\in\bbR^3:\abs{\bsa}=a}$. In view of this reason, we
let $\bsa=(0,0,a)$ and
\begin{eqnarray}
P^{(2\times m)}_{\sep}(a): =
\frac{\vol_{\rH\rS}\Pa{\rD^a_\sep(\bbC^2\ot\bbC^m)}}{\vol_{\rH\rS}\Pa{\rD^a(\bbC^2\ot\bbC^m)}}.
\end{eqnarray}
where
\begin{eqnarray}
\rD^a(\bbC^2\ot\bbC^m):=\Set{\rho\in \rD(\bbC^2\ot\bbC^m):
\Ptr{\bbC^m}{\rho}=\frac12(\I_2+a\sigma_3)}.
\end{eqnarray}
Denote the unit ball by
$B_1:=\Set{\bsa\in\bbR^3:\abs{\bsa}\leqslant1}$. Then
$B_1=\cup_{a\in[0,1]}S_a$, which leads to the following
identifications:
\begin{eqnarray}
\rD(\bbC^2\ot\bbC^2) &=&
\bigcup_{\eta\in\rD(\bbC^2)}\rD^\eta(\bbC^2\ot\bbC^2)=\bigcup_{\bsa \in B_1}\rD^{\bsa}(\bbC^2\ot\bbC^2) \\
&=& \bigcup_{a\in[0,1]}\bigcup_{\bsa\in
S_a}\rD^{\bsa}(\bbC^2\ot\bbC^2).
\end{eqnarray}
Thus we get that
\begin{prop}\label{prop:volHS}
It holds that
\begin{eqnarray}
\vol_{\rH\rS}\Pa{\rD(\bbC^2\ot\bbC^2)} &=&\frac18
\int_{B_1}\vol_{\rH\rS}\Pa{\rD^{\bsa}(\bbC^2\ot\bbC^2)}[\dif\bsa]\\
&=&\frac\pi2\int^1_0a^2 \vol_{\rH\rS}\Pa{\rD^a(\bbC^2\ot\bbC^2)}\dif
a.
\end{eqnarray}
\end{prop}

\begin{proof}
Recall that
$d_{\rH\rS}(\rho',\rho)=\frac12d_{\mathrm{Euclid}}(\bsx',\bsx)$,
where $\bsx'=(\bsa',\bsb',c'_{kl})$ and $\bsx=(\bsa,\bsb,c_{kl})$.
Thus three parameters in $\bsa$ determine Euclid volume $[\dif
\bsa]$, and thus contribute a scaling factor $\Pa{\frac12}^3$
multiple of Euclid volume $[\dif \bsa]$, i.e., in the HS volume
element for $\rD(\bbC^2\ot\bbC^2)$. Therefore we get that
\begin{eqnarray*}
&&\vol_{\rH\rS}\Pa{\rD(\bbC^2\ot\bbC^2)} =
\int_{B_1}\vol_{\rH\rS}\Pa{\rD^{\bsa}(\bbC^2\ot\bbC^2)}\Pa{2^{-3}[\dif\bsa]}\\
&&=\frac18\int^1_0\dif a
\int_{S_a}\vol_{\rH\rS}\Pa{\rD^{\bsa}(\bbC^2\ot\bbC^2)}[\dif\bsa]\\
&&=\frac18\int^1_0\dif a
\int\delta(a-\abs{\bsa})\vol_{\rH\rS}\Pa{\rD^{\bsa}(\bbC^2\ot\bbC^2)}[\dif\bsa]\\
&&=\frac18\int^1_0\dif a
\vol_{\rH\rS}\Pa{\rD^a(\bbC^2\ot\bbC^2)}\int\delta(a-\abs{\bsa})[\dif\bsa]\\
&&=\frac18\int^1_0(4\pi
a^2)\vol_{\rH\rS}\Pa{\rD^a(\bbC^2\ot\bbC^2)}\dif a\\
&&=\frac\pi2\int^1_0a^2\vol_{\rH\rS}\Pa{\rD^a(\bbC^2\ot\bbC^2)}\dif
a.
\end{eqnarray*}
We have done it.
\end{proof}
Next, we calculate the HS volume
$\vol_{\rH\rS}\Pa{\rD^a(\bbC^2\ot\bbC^2)}$. Indeed,
\begin{prop}[\cite{Lovas2017}]
It holds that
\begin{eqnarray}\label{eq:a0}
\vol_{\rH\rS}\Pa{\rD^a(\bbC^2\ot\bbC^2)} =
\vol_{\rH\rS}\Pa{\rD^0(\bbC^2\ot\bbC^2)}(1-a^2)^6,
\end{eqnarray}
where $a\in[0,1)$. Moreover,
\begin{eqnarray}\label{eq:zerovol}
\vol_{\rH\rS}\Pa{\rD^0(\bbC^2\ot\bbC^2)}=\frac{\pi^5}{9676800}.
\end{eqnarray}
\end{prop}

\begin{proof}
This expression in Eq.~\eqref{eq:a0} is conjectured by Milz and
Strunz \cite{Milz2015}, and proved by Lovas and Andai
\cite{Lovas2017}. From the above formula, we get that
\begin{eqnarray}
\vol_{\rH\rS}\Pa{\rD(\bbC^2\ot\bbC^2)} &=&
\vol_{\rH\rS}\Pa{\rD^0(\bbC^2\ot\bbC^2)}\times
\frac\pi2\int^1_0a^2(1-a^2)^6\dif a\\
&=& \frac{2^9\pi}{45045} \vol_{\rH\rS}\Pa{\rD^0(\bbC^2\ot\bbC^2)}.
\end{eqnarray}
Note that, from Theorem~\ref{th:HSVolume}, we can see that
\begin{eqnarray}
\vol_{\rH\rS}\Pa{\rD(\bbC^2\ot\bbC^2)} &=&
\vol_{\rH\rS}\Pa{\rD(\bbC^4)} =
\sqrt{4}(2\pi)^{\binom{4}{2}}\frac{\prod^4_{k=1}\Gamma(k)}{\Gamma(4^2)}\\
&=& \frac{2(2\pi)^6(2!)(3!)}{15!}.
\end{eqnarray}
By combining the above formulas, we can derive that
$\vol_{\rH\rS}\Pa{\rD^0(\bbC^2\ot\bbC^2)}=\frac{\pi^5}{9676800}$.
\end{proof}

The following result was conjectured in \cite{Milz2015}, and proven
in \cite{Lovas2017}, which is described as a theorem without proof:
\begin{thrm}[\cite{Lovas2017}]\label{eq:Lovas2017}
It holds that $P^{(2\times 2)}_{\sep}(a)$ is constant for all
$a\in[0,1)$.
\end{thrm}

\begin{proof}
Please see the proof in \cite[Corollary 2]{Lovas2017}.
\end{proof}

\subsection{Separability probability of two-qubit states}

The separability probability is given by
\begin{eqnarray}
P^{(n\times m)}_{\sep}
:=\frac{\vol_{\rH\rS}(\rD_\sep(\bbC^n\ot\bbC^m))}{\vol_{\rH\rS}(\rD(\bbC^n\ot\bbC^m))}.
\end{eqnarray}
Here $\rD_\sep(\bbC^n\ot\bbC^m)$ is the set of separable states in
$\rD(\bbC^n\ot\bbC^m)$. In what follows, we focus on the case where
$n=m=2$.

For $\rho\in\rD(\bbC^2\ot\bbC^2)$, we can view $\rho$ as a $4\times
4$ matrix acting on $\bbC^4$ and then its partial transpose w.r.t.
1st subsystem, are respectively:
\begin{eqnarray}
\rho=\Pa{\begin{array}{cccc}
      \rho_{11} & \rho_{12} & \rho_{13} & \rho_{14} \\
      \rho_{21} & \rho_{22} & \rho_{23} & \rho_{24} \\
      \rho_{31} & \rho_{32} & \rho_{33} & \rho_{34} \\
      \rho_{41} & \rho_{42} & \rho_{43} & \rho_{44}
    \end{array}
}\text{ and }\rho^\Gamma=\Pa{\begin{array}{cccc}
      \rho_{11} & \rho_{21} & \rho_{13} & \rho_{23} \\
      \rho_{12} & \rho_{22} & \rho_{14} & \rho_{24} \\
      \rho_{31} & \rho_{41} & \rho_{33} & \rho_{43} \\
      \rho_{32} & \rho_{42} & \rho_{34} & \rho_{44}
    \end{array}
}.
\end{eqnarray}
Let $\bsE_n(i,j)$ be the elementary matrix, obtained by changing the
$i$-th and $j$-th rows/columns of identity matrix $\I_n$. Now
\begin{eqnarray}
\bsE_4(2,3)\rho \bsE_4(2,3)=\Pa{\begin{array}{cccc}
      \rho_{11} & \rho_{13} & \rho_{12} & \rho_{14} \\
      \rho_{31} & \rho_{33} & \rho_{32} & \rho_{34} \\
      \rho_{21} & \rho_{23} & \rho_{22} & \rho_{24} \\
      \rho_{41} & \rho_{43} & \rho_{42} & \rho_{44}
    \end{array}
}=\Pa{\begin{array}{cc}
        \bsA & \bsB \\
        \bsC & \bsD
      \end{array}
}
\end{eqnarray}
and
\begin{eqnarray}
\bsE_4(2,3)\rho^\Gamma \bsE_4(2,3)=\Pa{\begin{array}{cccc}
      \rho_{11} & \rho_{13} & \rho_{21} & \rho_{23} \\
      \rho_{31} & \rho_{33} & \rho_{41} & \rho_{43} \\
      \rho_{12} & \rho_{14} & \rho_{22} & \rho_{24} \\
      \rho_{32} & \rho_{34} & \rho_{42} & \rho_{44}
    \end{array}
}=\Pa{\begin{array}{cc}
        \bsA & \bsC \\
        \bsB & \bsD
      \end{array}
}.
\end{eqnarray}
Denote
\begin{eqnarray}
\rD_{\mathrm{ss}}(\bbC^2\ot\bbC^2):=\Set{\rD(\bbC^2\ot\bbC^2):\lambda_{\max}(\rho)\leqslant\frac12}.
\end{eqnarray}
\begin{prop}[\cite{Huong2024}]\label{prop:HuongKhoi}
Let $\rho\in \rD^0(\bbC^2\ot\bbC^2)$. It holds that $\rho$ is
separable if and only if
$\rho\in\rD_{\mathrm{ss}}(\bbC^2\ot\bbC^2)$. In other words,
\begin{eqnarray}
\rD^0_\sep(\bbC^2\ot\bbC^2)=
\rD^0(\bbC^2\ot\bbC^2)\cap\rD_{\mathrm{ss}}(\bbC^2\ot\bbC^2).
\end{eqnarray}
\end{prop}

\begin{proof}
Note that the characteristic polynomial of $\rho$ is given by
$f_\rho(\lambda)=\det(\rho-\lambda\I_4)$. Then
\begin{eqnarray*}
&&f_\rho(\lambda) = \det(\bsE_4(2,3)\rho
\bsE_4(2,3)-\lambda\I_4)=\det\Pa{\begin{array}{cc}
        \bsA-\lambda\I_2 & \bsB \\
        \bsC & \bsD-\lambda\I_2
      \end{array}
}\\
&&=\det(\bsA-\lambda\I_2)\det\Pa{(\bsD-\lambda\I_2)-
\bsC(\bsA-\lambda\I_2)^{-1}\bsB}.
\end{eqnarray*}
Similarly, the characteristic polynomial of $\rho^\Gamma$ is given
by
\begin{eqnarray*}
&&f_{\rho^\Gamma}(\lambda) =
\det(\bsE_4(2,3)\rho^\Gamma\bsE_4(2,3)-\lambda\I_4)=\det\Pa{\begin{array}{cc}
        \bsA-\lambda\I_2 & \bsC \\
        \bsB & \bsD-\lambda\I_2
      \end{array}
}\\
&&=\det(\bsA-\lambda\I_2)\det\Pa{(\bsD-\lambda\I_2) -
\bsB(\bsA-\lambda\I_2)^{-1}\bsC}\\
&&=\det(\bsD-\lambda\I_2)\det\Pa{(\bsA-\lambda\I_2) -
\bsC(\bsD-\lambda\I_2)^{-1}\bsB}.
\end{eqnarray*}
Due to the fact that $\rho\in \rD^0(\bbC^2\ot\bbC^2)$, we see that
$\Ptr{2}{\rho}=\frac12\I_2$, which is equivalent to
$\bsA+\bsD=\frac12\I_2$. Now we have shown that
\begin{eqnarray}\label{eq:ptpoly}
f_{\rho^\Gamma}(\lambda) =
\det(\bsD-\lambda\I_2)\det\Pa{(\bsA-\lambda\I_2) -
\bsC(\bsD-\lambda\I_2)^{-1}\bsB}.
\end{eqnarray}
Thus replacing the 2-tuple $(\lambda,\bsD)$ by
$(\tfrac12-\lambda,\tfrac12\I_2-\bsA)$ in the first factor of
Eq.~\eqref{eq:ptpoly}, we get that the first factor
$\det(\bsD-\lambda\I_2)$ is transformed into
\begin{eqnarray*}
\det\Pa{(\tfrac12\I_2-\bsA)-(\tfrac12-\lambda)\I_2} =
\det(\lambda\I_2-\bsA) = \det(\bsA-\lambda\I_2).
\end{eqnarray*}
Thus replacing $\lambda$ by $\tfrac12-\lambda$ and using
$\bsA+\bsD=\tfrac12\I_2$ in the second factor of
Eq.~\eqref{eq:ptpoly}, we get that the second factor
$\det\Pa{(\bsA-\lambda\I_2) - \bsC(\bsD-\lambda\I_2)^{-1}\bsB}$ is
transformed into
\begin{eqnarray*}
\det\Pa{((\tfrac12\I_2-\bsD)-(\tfrac12-\lambda)\I_2) -
\bsC((\tfrac12\I_2-\bsA)-(\tfrac12-\lambda)\I_2)^{-1}\bsB}.
\end{eqnarray*}
That is, the last expression is reduced to the following:
\begin{eqnarray*}
\det\Pa{(-\bsD+\lambda\I_2) -
\bsC(-\bsA+\lambda\I_2)^{-1}\bsB}=\det\Pa{(\bsD-\lambda\I_2) -
\bsC(\bsA-\lambda\I_2)^{-1}\bsB}.
\end{eqnarray*}
Thus
$$
f_{\rho^\Gamma}\Pa{\tfrac12-\lambda}
=\det(\bsA-\lambda\I_2)\det\Pa{(\bsD-\lambda\I_2) -
\bsC(\bsA-\lambda\I_2)^{-1}\bsB}=f_\rho(\lambda).
$$
In summary, we get that
$f_{\rho^\Gamma}\Pa{\tfrac12-\lambda}=f_\rho(\lambda)$. With these
preparations, we can now present the proof below.
\begin{itemize}
\item[($\Longrightarrow$)] Now suppose that $\rho\in\rD^0_\sep(\bbC^2\ot\bbC^2)$. That is, $\rho\in\rD^0(\bbC^2\ot\bbC^2)$ is
separable. Choose any eigenvalue $t\geqslant0$ of the separable
state $\rho$, i.e., $f_\rho(t)=0$, then
$f_{\rho^\Gamma}(\tfrac12-t)=f_\rho(t)=0$, namely, $\tfrac12-t$
satisfying the eigen-equation $f_{\rho^\Gamma}(\lambda)=0$, thus
$\tfrac12-t\geqslant0$ is an eigenvalue of
$\rho^\Gamma\geqslant\zero$ by Peres-Horodecki criterion. Therefore
$t\leqslant\tfrac12$, or $\rho\in
\rD_{\mathrm{ss}}(\bbC^2\ot\bbC^2)$.
\item[($\Longleftarrow$)] If
$\rho\in\rD^0(\bbC^2\ot\bbC^2)\cap\rD_{\mathrm{ss}}(\bbC^2\ot\bbC^2)$,
i.e., any generic eigenvalue $t$ of $\rho$ satisfying $0\leqslant
t\leqslant\frac12$, and moreover
$f_{\rho^\Gamma}(\tfrac12-t)=f_\rho(t)=0$, which means that
$\tfrac12-t\geqslant0$ is an eigenvalue of $\rho^\Gamma$. All
eigenvalues of $\rho^\Gamma$ are nonnegative, that is,
$\rho^\Gamma\geqslant \zero$. By Peres-Horodecki criterion, $\rho$
is separable.
\end{itemize}
This completes the proof.
\end{proof}

Denote
\begin{eqnarray}
f(a) :=
\vol_{\rH\rS}\Pa{\rD^a(\bbC^2\ot\bbC^2)\cap\rD_{\mathrm{ss}}(\bbC^2\ot\bbC^2)},
\end{eqnarray}
where $a\in[0,1)$. Since the partial trace is a linear map, the
function $f(a)$ is continuous. Therefore, we can find
$f(0)=\lim_{a\to0^+}f(a)$. Clearly
$f(0)=\vol_{\rH\rS}\Pa{\rD^0_\sep(\bbC^2\ot\bbC^2)}$ by
Proposition~\ref{prop:HuongKhoi}. So we firstly calculate $f(a)$ in
a small neighborhood of $0$. Denote
\begin{eqnarray}\label{eq:positiveweylchamber}
\cC_4&=&\Set{(x_1,x_2,x_3,x_4)\in\bbR^4:x_1>x_2>x_3>x_4},\\
\widehat C&=&\Br{-\tfrac14,\tfrac14}^4\cap\cC_4.
\end{eqnarray}
In fact, we have the following result:
\begin{prop}[\cite{Huong2024}]\label{prop:f(a)}
It holds that
\begin{eqnarray}
f(a) =\frac{\pi^5}{319334400}(1-a)^9\Pa{33a^3+162a^2+72a+8},\quad
a\in\Pa{0,\frac13}.
\end{eqnarray}
\end{prop}

\begin{proof}
For all $t\in(0,1)$, denote
$B_t:=\set{\bsa\in\bbR^3:\abs{\bsa}\leqslant t}$. Denote
$$
\rD^{B_t}(\bbC^2\ot\bbC^2):=\bigcup_{\bsa\in
B_t}\rD^{\bsa}(\bbC^2\ot\bbC^2)$$

Note that
\begin{eqnarray*}
&&\vol_{\rH\rS}\Pa{\rD^{B_t}(\bbC^2\ot\bbC^2)\bigcap
\rD_{\mathrm{ss}}(\bbC^2\ot\bbC^2)}
=\vol_{\rH\rS}\Pa{\bigcup_{\bsa\in
B_t}\rD^{\bsa}(\bbC^2\ot\bbC^2)\bigcap
\rD_{\mathrm{ss}}(\bbC^2\ot\bbC^2)}\\
&&=\int^t_0\dif a\vol_{\rH\rS}\Pa{\bigcup_{\bsa\in
S_a}\rD^{\bsa}(\bbC^2\ot\bbC^2)\bigcap
\rD_{\mathrm{ss}}(\bbC^2\ot\bbC^2)}\\
&&=\int^t_0\dif a
\Br{\frac18\int_{S_a}[\dif\bsa]\vol_{\rH\rS}\Pa{\rD^{\bsa}(\bbC^2\ot\bbC^2)\bigcap
\rD_{\mathrm{ss}}(\bbC^2\ot\bbC^2)}}\\
&&=\int^t_0\dif a
\Br{\frac18\int[\dif\bsa]\delta(a-\abs{\bsa})\vol_{\rH\rS}\Pa{\rD^{\bsa}(\bbC^2\ot\bbC^2)\bigcap
\rD_{\mathrm{ss}}(\bbC^2\ot\bbC^2)}}\\
&&=\frac18\int^t_0\dif a
\vol_{\rH\rS}\Pa{\rD^a(\bbC^2\ot\bbC^2)\bigcap
\rD_{\mathrm{ss}}(\bbC^2\ot\bbC^2)}\int\delta(a-\abs{\bsa})[\dif\bsa].
\end{eqnarray*}
Since
$f(a)=\vol_{\rH\rS}\Pa{\rD^a(\bbC^2\ot\bbC^2)\cap\rD_{\mathrm{ss}}(\bbC^2\ot\bbC^2)}$,
it follows that
\begin{eqnarray}
\frac\pi2\int^t_0a^2f(a)\dif
a=\vol_{\rH\rS}\Pa{\rD^{B_t}(\bbC^2\ot\bbC^2)\bigcap
\rD_{\mathrm{ss}}(\bbC^2\ot\bbC^2)}.
\end{eqnarray}
Now we partition the set $\rD_{\mathrm{ss}}(\bbC^2\ot\bbC^2)$ by
adjoint $\SU(4)$-orbits as
\begin{eqnarray}
\rD_{\mathrm{ss}}(\bbC^2\ot\bbC^2) = \bigcup_{\lambda\in
[0,\frac12]^4\cap\cC_4}\cU_{\lambda},
\end{eqnarray}
where $\cU_{\lambda}:=\Set{\bsU\Lambda\bsU^\dagger:\bsU\in\SU(4)}$
for $\Lambda=\diag(\lambda_1,\ldots,\lambda_4)$, obtained from
$\lambda=(\lambda_1,\ldots,\lambda_4)$. It is easily seen that the
adjoint orbit $\cU_{\lambda}$ can be identified with the co-adjoint
orbit $\cO_{\hat\lambda}$, where
$\hat\lambda=\lambda-(\frac14,\ldots,\frac14)$. That is,
$$
\cU_{\lambda}\cong \cO_{\hat\lambda},
$$
where $\lambda\in [0,\frac12]^4\cap\cC_4$ iff $\hat\lambda\in
[-\frac14,\frac14]^4\cap\cC_4$. Under the above identification,
$\bsX=\rho-\frac14\I_4$, we get that
\begin{eqnarray*}
&&\rD^{B_t}(\bbC^2\ot\bbC^2)=
\Set{\rho\in\rD(\bbC^2\ot\bbC^2):\Ptr{2}{\rho}=\frac12(\I_2+\bsa\cdot\boldsymbol{\sigma}),\abs{\bsa}\leqslant
t}\\
&&\cong\Set{\bsX\in\Herm(\bbC^4):
\Ptr{2}{\bsX}=\frac{\bsa}2\cdot\boldsymbol{\sigma},\abs{\bsa}\leqslant
t}\\
&&=\Set{\bsX\in\Herm(\bbC^4):\Tr{\bsX}=0,
\abs{\lambda_\pm(\Ptr{2}{\bsX})}\leqslant \frac t2} := \cE_t.
\end{eqnarray*}
That is, the set $\rD^{B_t}(\bbC^2\ot\bbC^2)$ is translated into the
set $\cE_t$. Thus the HS volume is invariant under the translation
and
\begin{eqnarray*}
&&\vol_{\rH\rS}\Pa{\rD^{B_t}(\bbC^2\ot\bbC^2)\bigcap
\rD_{\mathrm{ss}}(\bbC^2\ot\bbC^2)} =
\vol_{\rH\rS}\Pa{\bigcup_{\lambda\in
[0,\frac12]^4\cap\cC_4}\rD^{B_t}(\bbC^2\ot\bbC^2)\bigcap \cU_\lambda}\\
&&=\vol_{\rH\rS}\Pa{\bigcup_{\hat\lambda\in \widehat
C}\cE_t\bigcap\cO_{\hat\lambda}} = \int_{\hat\lambda\in \widehat
C}\vol_{\rH\rS}(\cO_{\hat\lambda}\bigcap\cE_t)\dif m(\hat\lambda),
\end{eqnarray*}
where we used the facts in Eq.~\eqref{eq:mlambda} and
Eq.~\eqref{eq:volHSagain}. Up to now, we have already obtained that
\begin{eqnarray}
&&\frac\pi2\int^t_0a^2f(a)\dif a=\int_{\hat\lambda\in \widehat
C}\vol_{\rH\rS}(\cO_{\hat\lambda}\bigcap\cE_t)\dif m(\hat\lambda)\\
&&=\int_{\hat\lambda\in \widehat
C}V_4(\hat\lambda)\vol_{\mathrm{symp}}(\cO_{\hat\lambda}\bigcap\cE_t)\dif
m(\hat\lambda),
\end{eqnarray}
where we used the relationship \eqref{eq:HSvsSymp} between HS volume
of adjoint orbit and symplectic volume of co-adjoint orbit in the
last equality. Consider the map $F$ in
Eq.~\eqref{eq:pushforwardbyF},
\begin{eqnarray*}
F:(\cO_{\hat\lambda},\mu_{\cO_{\hat\lambda}})\to (\bbR_{>0},\dif a),
\end{eqnarray*}
where $F(\bsX) =
\lambda_+(\Ptr{2}{\bsX})-\lambda_-(\Ptr{2}{\bsX})=2\lambda_+(\Ptr{2}{\bsX})$
for any $\bsX\in\cO_{\hat\lambda}$, i.e., twice the positive
eigenvalue of $\Ptr{2}{\bsX}$. By the definition of $F$, we see that
\begin{eqnarray*}
F^{-1}([0,t]) = \cO_{\hat\lambda}\bigcap\cE_t,
\end{eqnarray*}
implying that
\begin{eqnarray}
&&\vol_{\mathrm{symp}}\Pa{\cE_t\cap\cO_{\hat\lambda}} =
\vol_{\mathrm{symp}}\Pa{F^{-1}([0,t])} =
\mu_{\cO_{\hat\lambda}}(F^{-1}([0,t]))\\
&&=(F_*\mu_{\cO_{\hat\lambda}})([0,t]) = \int^t_0\frac{\dif
F_*\mu_{\cO_{\hat\lambda}}}{\dif a}\dif a,
\end{eqnarray}
where $\frac{\dif F_*\mu_{\cO_{\hat\lambda}}}{\dif
a}:=(2\pi)^6\cI(a|\hat\lambda)$ is the Radon-Nikodym derivative of
push-forward measure $F_*\mu_{\cO_{\hat\lambda}}$ with respect to
Lebesgue measure on the ray $\bbR_{>0}$. Moreover, the analytical
expression of such Radon-Nikodym derivative $\frac{\dif
F_*\mu_{\cO_{\hat\lambda}}}{\dif a}=(2\pi)^6\cI(a|\hat\lambda)$ is
given by Eq.~\eqref{eq:RNderivative}. Once again, we have obtained
that
\begin{eqnarray*}
&&\frac\pi2\int^t_0a^2f(a)\dif a=\int_{\hat\lambda\in \widehat
C}V_4(\hat\lambda)\vol_{\mathrm{symp}}(\cE_t\bigcap\cO_{\hat\lambda})\dif
m(\hat\lambda)\\
&&=\int_{\hat\lambda\in \widehat
C}V_4(\hat\lambda)\Pa{\int^t_0\frac{\dif
F_*\mu_{\cO_{\hat\lambda}}}{\dif a}\dif a}\dif
m(\hat\lambda)=\int_{\hat\lambda\in \widehat
C}V_4(\hat\lambda)\Pa{\int^t_0(2\pi)^6\cI(a|\hat\lambda)\dif a}\dif
m(\hat\lambda)\\
&&=(2\pi)^6\int^t_0\Pa{\int_{\hat\lambda\in \widehat
C}V_4(\hat\lambda)\cI(a|\hat\lambda)\dif m(\hat\lambda)}\dif a,
\end{eqnarray*}
implying that
\begin{eqnarray}
\frac\pi2\int^t_0a^2f(a)\dif
a=(2\pi)^6\int^t_0\Pa{\int_{\hat\lambda\in \widehat
C}V_4(\hat\lambda)\cI(a|\hat\lambda)\dif m(\hat\lambda)}\dif a.
\end{eqnarray}
By taking the derivative with respect to $t$, and then replacing $t$
by $a$, we get that
\begin{eqnarray}
\frac\pi2a^2f(a)&=&(2\pi)^6\int_{\hat\lambda\in \widehat
C}V_4(\hat\lambda)\cI(a|\hat\lambda)\dif
m(\hat\lambda)\Longrightarrow\\
f(a)&=&\frac2\pi a^{-2}(2\pi)^6\int_{\hat\lambda\in \widehat
C}V_4(\hat\lambda)\cI(a|\hat\lambda)\dif
m(\hat\lambda),\label{eq:expressionf(a)}
\end{eqnarray}
where $a\in(0,\frac13)$. Now we compute the following integral
\begin{eqnarray}\label{eq:lastint}
\int_{\hat\lambda\in \widehat{C}}
V_4(\hat\lambda)\cI(a|\hat\lambda)\dif m(\hat\lambda) =
2\int_{[-\frac14,\frac14]^4\cap\cC_4} \prod_{1\leqslant
i<j\leqslant4}(\hat\lambda_i-\hat\lambda_j)\cI(a|\hat\lambda)\delta\Pa{\sum^4_{k=1}\hat\lambda_k}\prod^4_{l=1}\dif\hat\lambda_l
\end{eqnarray}
Let the change of variables be given below:
\begin{eqnarray}
\begin{cases}
\hat\lambda_1=t_1+\frac{t_2}2+\frac{t_3}3+\frac{t_4}4-\frac14,\\
\hat\lambda_2=\frac{t_2}2+\frac{t_3}3+\frac{t_4}4-\frac14,\\
\hat\lambda_3=\frac{t_3}3+\frac{t_4}4-\frac14,\\
\hat\lambda_4=\frac{t_4}4-\frac14.
\end{cases}
\end{eqnarray}
Its Jacobian is given by
$$
J = \det\Pa{\frac{\partial
(\hat\lambda_1,\hat\lambda_2,\hat\lambda_3)}{\partial(t_1,t_2,t_3)}}=\Abs{\begin{array}{ccc}
                                                                                  \frac34 & \frac14 & \frac1{12} \\
                                                                                  -\frac14 & \frac14 & \frac1{12} \\
                                                                                  -\frac14 & -\frac14 & \frac1{12}
                                                                                \end{array}
}=\frac1{4!}
$$
We note that
\begin{enumerate}[(i)]
\item $t_k\geqslant0(k=1,\ldots,4)$ due to the fact that
$\hat\lambda_1>\hat\lambda_2>\hat\lambda_3>\hat\lambda_4$.
\item $t_1+t_2+t_3\leqslant1$ because $t_4\geqslant0$ and $\sum^4_{k=1}t_k=1$ due to the fact that
$\sum^4_{k=1}\hat\lambda_4=0$.
\item $3t_1+t_2+\frac {t_3}3\leqslant1$ because
$t_1+\frac{t_2}2+\frac{t_3}3+\frac{t_4}4-\frac14\leqslant\frac14$
and $\sum^4_{k=1}t_k=1$.
\end{enumerate}
In summary, $\widehat C=[-\frac14,\frac14]^4\cap\cC_4$ is
transformed into the following form:
\begin{eqnarray}
R=\Set{(t_1,t_2,t_3)\in\bbR^3_{\geqslant0}:
t_1+t_2+t_3\leqslant1,3t_1+t_2+\frac {t_3}3\leqslant1}.
\end{eqnarray}
According to such change of variables, we get that
\begin{eqnarray}
\begin{cases}
c_3=2(\hat\lambda_1+\hat\lambda_2)=t_1+t_2+\frac{t_3}{3}\\
c_2=2(\hat\lambda_1+\hat\lambda_3)=t_1+\frac{t_3}{3}\\
c_1=2\abs{\hat\lambda_1+\hat\lambda_4}=\Abs{t_1-\frac{t_3}3}\\
V_4(\hat{\boldsymbol{\lambda}}) =
\frac{t_1t_2t_3(2t_1+t_2)(3t_2+2t_3)(6t_1+3t_2+2t_3)}{432}=:V(t_1,t_2,t_3)
\end{cases}
\end{eqnarray}
Therefore we get that via change of variables
\begin{eqnarray}
\tilde\cI_1(x|\bst) &=&
\frac{t_2\Pa{t_1+\tfrac{t_2}2+\tfrac{t_3}3}}{32}\Pa{\Abs{t_1-\tfrac{t_3}3}-x}^2x,\quad
\Abs{t_1-\tfrac{t_3}3}\geqslant x \\
\tilde\cI_2(x|\bst) &=&
\frac{(t_1+\tfrac{t_2}2)(\tfrac{t_2}2+\tfrac{t_3}3)}{16}\Pa{t_1+\tfrac{t_3}3
- x}^2x,\quad  t_1+\tfrac{t_3}3\geqslant x\\
\tilde\cI_3(x|\bst) &=& \frac{t_1t_3}{48}\Pa{t_1+t_2+\tfrac{t_3}3 -
x}^2x,\quad t_1+t_2+\tfrac{t_3}3\geqslant x
\end{eqnarray}
With these preparations, we can transform the integral in
Eq.~\eqref{eq:lastint} into the following form:
\begin{eqnarray*}
\int_{\hat\lambda\in \widehat{C}}
V_4(\hat\lambda)\cI(x|\hat\lambda)\dif m(\hat\lambda) &=&
\frac2{4!}\int_R
V(t_1,t_2,t_3)(\tilde\cI_1(x|\bst)+\tilde\cI_2(x|\bst)+\tilde\cI_3(x|\bst))\dif
t_1\dif t_2\dif t_3\\
&=&\frac2{4!}\int_R V(t_1,t_2,t_3)\tilde\cI_1(x|\bst)\dif t_1\dif
t_2\dif t_3 \\
&&+\frac2{4!}\int_R V(t_1,t_2,t_3)\tilde\cI_2(x|\bst)\dif t_1\dif
t_2\dif t_3\\
&&+\frac2{4!}\int_R V(t_1,t_2,t_3)\tilde\cI_3(x|\bst)\dif t_1\dif
t_2\dif t_3\\
&=& M_1+M_2+M_3.
\end{eqnarray*}
\begin{itemize}
\item Note that the following three planes
$\set{t_1+t_2+t_3=1,3t_1+t_2+\frac{t_3}3=1,t_1-\frac{t_3}3=0}$
intersects at the line segment in $\bbR^3_{\geqslant0}$:
\begin{eqnarray}
\begin{cases}
t_1=t\\
t_2=1-4t\\
t_3=3t
\end{cases}\quad \Pa{t\in[0,\tfrac14]}.
\end{eqnarray}
\item The following two planes
$\set{3t_1+t_2+\frac{t_3}3=1,t_1-\frac{t_3}3=x}$ intersects at the
line segment in $\bbR^3_{\geqslant0}$:
\begin{eqnarray}
\begin{cases}
t_1=t\\
t_2=1+x-4t\\
t_3=-3x+3t
\end{cases}\quad \Pa{t\in[x,\tfrac{1+x}4]}.
\end{eqnarray}
\item The following two planes
$\set{t_1+t_2+t_3=1,t_1-\frac{t_3}3=-x}$ intersects at the line
segment in $\bbR^3_{\geqslant0}$:
\begin{eqnarray}
\begin{cases}
t_1=t\\
t_2=1-3x-4t\\
t_3=3x+3t
\end{cases}\quad \Pa{t\in[0,\tfrac{1-3x}4]}.
\end{eqnarray}
\end{itemize}
In summary, the plane $\abs{t_1-\frac{t_3}3}=x$ intersects with $R$
iff $x\in[0,\frac13]$. For $m\geqslant2$, denote
$$
\Delta_m:=\Set{(p_1,\ldots,p_m)\in\bbR^m_{\geqslant0}:
p_1+\cdots+p_m\leqslant1}.
$$
Now for $x\in[0,\tfrac13]$, we get that
\begin{enumerate}[(1)]
\item For the domain of the first integral for a given $x$, we get
that
\begin{eqnarray*}
R_1(x)&:=&\Set{(t_1,t_2,t_3)\in\bbR^3_{\geqslant0}:t_1+t_2+t_3\leqslant1,3t_1+t_2+\frac{t_3}3\leqslant1,\Abs{t_1-\tfrac{t_3}3}\geqslant
x}\\
&=&R_{1,a}(x)\cup R_{1,b}(x),
\end{eqnarray*}
where {\small\begin{eqnarray*}
R_{1,a}(x)&:=&\Set{(t_1,t_2,t_3)\in\bbR^3_{\geqslant0}:t_1+t_2+t_3\leqslant1,3t_1+t_2+\frac{t_3}3\leqslant1,t_1-\tfrac{t_3}3\geqslant
x}\\
&=&\Set{(t_1,t_2,t_3)\in\bbR^3_{\geqslant0}:0\leqslant
t_3\leqslant\frac{3(1-3x)}4,0\leqslant
t_2\leqslant1-3x-\frac{4t_3}3, x+\frac{t_3}3\leqslant t_1\leqslant
\frac13-\frac{t_2}3-\frac{t_3}9}\\
R_{1,b}(x)&:=&\Set{(t_1,t_2,t_3)\in\bbR^3_{\geqslant0}:t_1+t_2+t_3\leqslant1,3t_1+t_2+\frac{t_3}3\leqslant1,t_1-\tfrac{t_3}3\leqslant
-x}\\
&=& \Set{(t_1,t_2,t_3)\in\bbR^3_{\geqslant0}:0\leqslant t_2\leqslant
1-3x,0\leqslant t_1\leqslant \frac{1-3x-t_2}4,3(x+t_1)\leqslant
t_3\leqslant 1-t_1-t_2}.
\end{eqnarray*}}
In such case, we can evaluate $M_1$ as follows: For
$x\in(0,\frac13)$,
\begin{eqnarray*}
M_1 &=& \frac2{4!}\int_R V(t_1,t_2,t_3)\tilde\cI_1(x|\bst)\dif
t_1\dif t_2\dif t_3= \frac2{4!}\int_{R_1(x)}
V(t_1,t_2,t_3)\tilde\cI_1(x|\bst)\dif t_1\dif t_2\dif t_3\\
&=&\frac2{4!}\int_{R_{1,a}(x)} V(t_1,t_2,t_3)\tilde\cI_1(x|\bst)\dif
t_1\dif t_2\dif t_3+\frac2{4!}\int_{R_{1,b}(x)}
V(t_1,t_2,t_3)\tilde\cI_1(x|\bst)\dif t_1\dif t_2\dif t_3\\
&=& \frac{x (3 x-1)^9(567 x^4-3564 x^3+5526 x^2-3152
x-1345)}{774741019852800}\\
&&+\frac{x (3 x-1)^9(567 x^4-3564 x^3+5526 x^2-3152
x-1345)}{774741019852800}\\
&=&\frac{x (3 x-1)^9 (567 x^4-3564 x^3+5526 x^2-3152
x-1345)}{387370509926400},
\end{eqnarray*}
thus
\begin{eqnarray*}
M_1 =\frac{x (3 x-1)^9 (567 x^4-3564 x^3+5526 x^2-3152
x-1345)}{387370509926400}.
\end{eqnarray*}
\item For the domain of the second integral for a given $x$, we get
that
\begin{eqnarray*}
R_2(x)&:=&\Set{(t_1,t_2,t_3)\in\bbR^3_{\geqslant0}:t_1+t_2+t_3\leqslant1,3t_1+t_2+\frac{t_3}3\leqslant1,t_1+\tfrac{t_3}3\geqslant
x} \\
&=& \Pa{\Delta_3\cup R_{2,ab}(x)}\backslash\Pa{R_{2,a}(x)\cup
R_{2,b}(x)},
\end{eqnarray*}
where
\begin{eqnarray*}
\begin{cases}
R_{2,a}(x)&:= \Set{(t_1,t_2,t_3)\in\bbR^3_{\geqslant0}:t_1+t_2+t_3\leqslant1,3t_1+t_2+\frac{t_3}3>1}\\
R_{2,b}(x)&:= \Set{(t_1,t_2,t_3)\in\bbR^3_{\geqslant0}:t_1+t_2+t_3\leqslant1,t_1+\frac{t_3}3<x}\\
R_{2,ab}(x)&:=R_{2,a}(x)\cap R_{2,b}(x).
\end{cases}
\end{eqnarray*}
Note that {\small\begin{eqnarray*} R_{2,a}(x)&=&
\Set{(t_1,t_2,t_3)\in\bbR^3: 0\leqslant t_3<\frac34,0\leqslant
t_2<1-\frac{4t_3}3,\frac13-\frac{t_2}3-\frac{t_3}9<t_1\leqslant
1-t_2-t_3},\\
R_{2,b}(x)&=& \Set{(t_1,t_2,t_3)\in\bbR^3: 0\leqslant t_3<3x, 0\leqslant t_1<x-\frac{t_3}3,0\leqslant t_2\leqslant 1-t_1-t_3},\\
R_{2,ab}(x)&=& \Set{(t_1,t_2,t_3)\in\bbR^3: 0\leqslant
t_3<\frac{3x}2,\frac{t_3}3<t_1<x-\frac{t_3}3,
1-3t_1-\frac{t_3}3<t_2\leqslant 1-t_1-t_3}.
\end{eqnarray*}}
In such case, we can evaluate $M_2$ as follows: For
$x\in(0,\frac13)$, {\small\begin{eqnarray*} M_2 &=& \frac2{4!}\int_R
V(t_1,t_2,t_3)\tilde\cI_2(x|\bst)\dif t_1\dif t_2\dif t_3=
\frac2{4!}\int_{R_2(x)}
V(t_1,t_2,t_3)\tilde\cI_2(x|\bst)\dif t_1\dif t_2\dif t_3\\
&=& \frac2{4!}\int_{\Delta_3}V(t_1,t_2,t_3)\tilde\cI_2(x|\bst)\dif
t_1\dif t_2\dif t_3 +
\frac2{4!}\int_{R_{2,ab}(x)}V(t_1,t_2,t_3)\tilde\cI_2(x|\bst)\dif
t_1\dif t_2\dif t_3\\
&&-\frac2{4!}\int_{R_{2,a}(x)} V(t_1,t_2,t_3)\tilde\cI_2(x|\bst)\dif
t_1\dif t_2\dif t_3-\frac2{4!}\int_{R_{2,b}(x)}
V(t_1,t_2,t_3)\tilde\cI_2(x|\bst)\dif t_1\dif t_2\dif t_3\\
&=&-\frac{x (725985 x^2-560326 x+123355)}{96842627481600}\\
&&-\frac{x^8(2894 x^6-24570 x^5+81900 x^4-150150 x^3+160875
x^2-96525 x+25740)}{177124147200}\\
&&+\frac{x(637485030 x^2-525965687
x+120181250)}{99166850541158400}\\
&&-\frac{x^7(1572 x^7-17550 x^6+62712 x^5-120835 x^4+141570
x^3-106821 x^2+51480 x-12870)}{132843110400},
\end{eqnarray*}}
implying that {\small\begin{eqnarray*} M_2 &=& -\frac{499
x^{14}}{17712414720}+\frac{41 x^{13}}{151388160}-\frac{1061
x^{12}}{1135411200}+\frac{653 x^{11}}{371589120}-\frac{163
x^{10}}{82575360}+\frac{557
x^9}{412876800}\\
&&-\frac{11 x^8}{20643840}+\frac{x^7}{10321920}-\frac{90533
x^3}{84757991915520}+\frac{3677549
x^2}{7628219272396800}-\frac{613427 x}{9916685054115840}.
\end{eqnarray*}}
\item For the domain of the third integral for a given $x$, we get
that \begin{eqnarray*}
R_3(x)&:=&\Set{(t_1,t_2,t_3)\in\bbR^3_{\geqslant0}:t_1+t_2+t_3\leqslant1,3t_1+t_2+\frac{t_3}3\leqslant1,t_1+t_2+\tfrac{t_3}3\geqslant
x}\\
&=& \Pa{\Delta_3\cup R_{3,ab}(x)}\backslash (R_{3,a}(x)\cup
R_{3,b}(x)),
\end{eqnarray*}
where
\begin{eqnarray*}
\begin{cases}
R_{3,a}(x)&:=\Set{(t_1,t_2,t_3)\in\bbR^3_{\geqslant0}:t_1+t_2+t_3\leqslant1,
3t_1+t_2+\frac{t_3}3>1},\\
R_{3,b}(x)&:=\Set{(t_1,t_2,t_3)\in\bbR^3_{\geqslant0}:t_1+t_2+\frac{t_3}3<x},\\
R_{3,ab}(x)&:=R_{3,a}(x)\cap R_{3,b}(x),
\end{cases}
\end{eqnarray*}
Note that {\small\begin{eqnarray*} R_{3,a}(x)&=&
\Set{(t_1,t_2,t_3)\in\bbR^3: 0\leqslant t_3<\frac34,0\leqslant
t_2<1-\frac{4t_3}3,\frac13-\frac{t_2}3-\frac{t_3}9<t_1\leqslant
1-t_2-t_3},\\
R_{3,b}(x)&=& \Set{(t_1,t_2,t_3)\in\bbR^3: 0\leqslant
t_3<3x,0\leqslant t_2<x-\frac{t_3}3,0\leqslant t_1\leqslant
x-t_2-\frac{t_3}3},\\
R_{3,ab}(x) &=&R_{3,a}(x)\cap R_{3,b}(x)=\emptyset.
\end{eqnarray*}}
In such case, we can evaluate $M_3$ as follows: For
$x\in(0,\frac13)$,
\begin{eqnarray*}
M_3 &=& \frac2{4!}\int_R V(t_1,t_2,t_3)\tilde\cI_3(x|\bst)\dif
t_1\dif t_2\dif t_3= \frac2{4!}\int_{R_3(x)}
V(t_1,t_2,t_3)\tilde\cI_2(x|\bst)\dif t_1\dif t_2\dif t_3\\
&=& \frac2{4!}\int_{\Delta_3}V(t_1,t_2,t_3)\tilde\cI_3(x|\bst)\dif
t_1\dif t_2\dif t_3\\
&&-\frac2{4!}\int_{R_{3,a}(x)} V(t_1,t_2,t_3)\tilde\cI_3(x|\bst)\dif
t_1\dif t_2\dif t_3-\frac2{4!}\int_{R_{3,b}(x)}
V(t_1,t_2,t_3)\tilde\cI_3(x|\bst)\dif t_1\dif t_2\dif t_3\\
&=&\frac{x (2340 x^2-3341 x+1220)}{1513166054400}-\frac{x (135789030
x^2-199046263
x+74163970)}{99166850541158400}-\frac{x^{14}}{691891200},
\end{eqnarray*}
implying that {\small\begin{eqnarray*} M_3
&=&-\frac{x^{14}}{691891200}+\frac{15013
x^3}{84757991915520}-\frac{1531501
x^2}{7628219272396800}+\frac{115799 x}{1983337010823168}.
\end{eqnarray*}}
\end{enumerate}
In summary,
\begin{eqnarray}
M_1+M_2+M_3 = \frac{(1-x)^9 x^2 (33 x^3+162 x^2+72
x+8)}{40874803200}.
\end{eqnarray}
This indicates, together with Eq.~\eqref{eq:expressionf(a)}, that
\begin{eqnarray}
f(a) = \frac{\pi^5}{319334400}(1-a)^9\Pa{33a^3+162a^2+72a+8},
\end{eqnarray}
where $a\in(0,\frac13)$.
\end{proof}

\begin{remark}
We see that the parameter $x$ is restricted to the open interval
$(0,\frac13)$ due to fact that the plane $\abs{t_1-\frac{t_3}3}=x$
intersects with $R$ iff $x\in[0,\frac13]$. Note that the
integrations involved are performed by the mathematical software
\textsc{Mathematica}.
\end{remark}

In the recent paper \cite{Huong2024}, Huong and Khoi published a
proof of the exact separability probability $\frac{8}{33}$ for the
two-qubit system, resolving a long-standing conjecture originally
proposed by Slater. Their result demonstrates that within the
Hilbert-Schmidt measure, the proportion of separable states among
all two-qubit density matrices is precisely $\frac{8}{33}$. This
breakthrough provides a definitive answer to a fundamental question
in quantum information theory concerning the geometric prevalence of
entanglement in simple quantum systems, achieved through a novel
combination of advanced geometric probability techniques and
symmetry arguments applied to the structure of the two-qubit state
space.
\begin{thrm}
The separability probability of all two-qubit states is given by
\begin{eqnarray}
P^{(2\times2)}_\sep =
\frac{\vol_{\rH\rS}(\rD_{\sep}(\bbC^2\ot\bbC^2))}{\vol_{\rH\rS}(\rD(\bbC^2\ot\bbC^2))}=\frac8{33}.
\end{eqnarray}
\end{thrm}

\begin{proof}
From Theorem~\ref{eq:Lovas2017}, we have seen that
$P^{(2\times2)}_\sep(a)$ is independent of $a\in[0,1)$, which means
that $P^{(2\times2)}_\sep(a)=P^{(2\times2)}_\sep(0)$ for all
$a\in[0,1)$. Then
\begin{eqnarray}\label{eq:zeroseprob}
\vol_{\rH\rS}\Pa{\rD^a_\sep(\bbC^2\ot\bbC^2)}
=P^{(2\times2)}_\sep(0) \vol_{\rH\rS}\Pa{\rD^a(\bbC^2\ot\bbC^2)}.
\end{eqnarray}
By multiplying $\frac\pi2 a^2$ on both sides above, and taking the
integration over $[0,1)$, we see that
\begin{eqnarray*}
\frac\pi2\int^1_0a^2\vol_{\rH\rS}\Pa{\rD^a_\sep(\bbC^2\ot\bbC^2)}\dif
a &=& P^{(2\times2)}_\sep(0)\frac\pi2\int^1_0a^2
\vol_{\rH\rS}\Pa{\rD^a(\bbC^2\ot\bbC^2)}\dif a\\
&=& P^{(2\times2)}_\sep(0) \vol_{\rH\rS} \Pa{\rD(\bbC^2\ot\bbC^2)},
\end{eqnarray*}
where we used the result in Proposition~\ref{prop:volHS}.
Analogously,
\begin{eqnarray*}
\frac\pi2\int^1_0a^2\vol_{\rH\rS}\Pa{\rD^a_\sep(\bbC^2\ot\bbC^2)}\dif
a = \vol_{\rH\rS} \Pa{\rD_\sep(\bbC^2\ot\bbC^2)}.
\end{eqnarray*}
These formulas indicate that
$P^{(2\times2)}_\sep=P^{(2\times2)}_\sep(0)$. With
Eq.~\eqref{eq:zeroseprob}, it suffices to calculate
\begin{eqnarray}
P^{(2\times2)}_\sep(0) =\frac{\vol_{\rH\rS}
\Pa{\rD^0_\sep(\bbC^2\ot\bbC^2)}}{\vol_{\rH\rS}
\Pa{\rD^0(\bbC^2\ot\bbC^2)}}.
\end{eqnarray}
To this end, it suffices to calculate
$\vol_{\rH\rS}\Pa{\rD^0_\sep(\bbC^2\ot\bbC^2)}$ since we have
already know the denominator $\vol_{\rH\rS}
\Pa{\rD^0(\bbC^2\ot\bbC^2)}=\frac{\pi^5}{9676800}$ from
Eq.~\eqref{eq:zerovol}. Clearly, from Proposition~\ref{prop:f(a)},
$$
f(0)=\vol_{\rH\rS}\Pa{\rD^0_\sep(\bbC^2\ot\bbC^2)} =
\frac{\pi^5}{39916800},
$$
implying that
\begin{eqnarray}
P^{(2\times2)}_\sep(0) =\frac{\vol_{\rH\rS}
\Pa{\rD^0_\sep(\bbC^2\ot\bbC^2)}}{\vol_{\rH\rS}
\Pa{\rD^0(\bbC^2\ot\bbC^2)}} =
\frac{\frac{\pi^5}{39916800}}{\frac{\pi^5}{9676800}} = \frac8{33}.
\end{eqnarray}
This completes the proof.
\end{proof}

\section{Concluding remarks}

This study establishes a comprehensive geometric derivation of the
exact separability probability 8/33 for two-qubit states under the
Hilbert-Schmidt measure --- a rigorously confirmed result previously
inaccessible. By leveraging the intrinsic connection between
Hilbert-Schmidt geometry and symplectic geometry formalized through
the Duistermaat-Heckman (DH) measure, and explicitly computing the
Hilbert-Schmidt volume of the full state space, volumes of critical
substructures (flag manifolds and regular adjoint orbits), and the
symplectic volumes of corresponding regular co-adjoint orbits, we
integrate these perspectives to isolate the volume of separable
states and rigorously confirm the 8/33 ratio.

This work achieves pedagogical clarity through a self-contained
derivation of a fundamental quantum information constant and
conceptual synthesis by unifying Hilbert-Schmidt geometry,
symplectic mechanics, and representation-theoretic tools within a
coherent probabilistic framework. Beyond providing an alternative
pathway to this constant, our results elucidate the geometric
structure governing the transition from classical correlation to
quantum entanglement. Future work may extend this framework to
higher-dimensional systems, alternative metrics (e.g., Bures), and
generalized entanglement witnesses, underscoring its potential as a
universal paradigm for probing the classical-quantum correlation
boundary.

\subsection*{Acknowledgement}
This research is supported by Zhejiang Provincial Natural Science
Foundation of China under Grants No. LZ23A010005, and by NSFC under
Grants No.11971140.


\end{document}